\documentclass[10pt]{article}
\usepackage[sectionbib]{natbib}
\usepackage{array,epsfig,fancyhdr,rotating}
\usepackage[dvipdfm]{hyperref}
\usepackage{geometry}
\usepackage{amsmath}
\usepackage{amssymb}
\usepackage{amsfonts}
\usepackage{multirow}
\usepackage{amsthm}

\newtheorem{theorem}{Theorem}
\newtheorem{lemma}{Lemma}

\newtheorem{proposition}{Proposition}
\theoremstyle{definition}

\begin{document}


\renewcommand{\baselinestretch}{1.2}

\markright{ \hbox{\footnotesize\rm Statistica Sinica
}\hfill\\[-13pt]
\hbox{\footnotesize\rm
}\hfill }

\markboth{\hfill{\footnotesize\rm SIMENG QU AND XIAO WANG} \hfill}
{\hfill {\footnotesize\rm FILL IN A SHORT RUNNING TITLE} \hfill}

\renewcommand{\thefootnote}{}
$\ $\par


\fontsize{10.95}{14pt plus.8pt minus .6pt}\selectfont
\vspace{0.8pc}
\centerline{\large\bf Optimal Global Test for Functional Regression}
\vspace{2pt}
\centerline{Simeng Qu and Xiao Wang}
\vspace{.4cm}
\centerline{\it Department of Statistics, Purdue University}
\vspace{.55cm}
\fontsize{9}{11.5pt plus.8pt minus .6pt}\selectfont


\begin{quotation}
\noindent {\it Abstract:}
This paper studies the optimal testing for the nullity of the slope function
in the functional linear model using smoothing splines. 
We propose a generalized likelihood ratio test
based on an easily implementable data-driven estimate. 
The quality of the test is measured by the minimal distance
between the null and the alternative set that still allows a possible test. 
The lower bound of the minimax decay rate of this distance is derived, and test with a distance that decays faster than the lower bound would be impossible.
We show that the minimax optimal rate is
jointly determined by the smoothing spline kernel and the covariance kernel.
It is shown that our test attains this optimal
rate. 
Simulations are carried out to confirm the finite-sample performance of our test as well as to illustrate the theoretical results. Finally, we apply our test to study the effect
of the trajectories of oxides of nitrogen ($\text{NO}_{\text{x}}$) on the level of ozone
($\text{O}_3$).\par

\vspace{9pt}
\par
\end{quotation}\par

\def\thefigure{\arabic{figure}}
\def\thetable{\arabic{table}}

\fontsize{10.95}{14pt plus.8pt minus .6pt}\selectfont


\section{Introduction}

Functional linear regression model, with respect to nonparametric estimation
and prediction, has drawn extensive attention in the field of functional data analysis in recent years. The model is stated as 
\begin{equation}
Y=a_{0}+\int_{0}^{1}\beta_{0}(t)X(t)dt+\epsilon,\label{equ:flr}
\end{equation}
where $Y$ is a scalar response, $X:[0,1]\rightarrow\mathbb{R}$ is
a square integrable random functional predictor, $\alpha_{0}\in\mathbb{R}$
is the intercept, $\beta_{0}:[0,1]\rightarrow\mathbb{R}$ is the slope
function, and $\epsilon$ is the random error with mean zero and variance
$\sigma^{2}$. 
One of the popular methods to study such model is based on the functional
principal component analysis \citep{james_02,ramsay_05,yao_05,cai_06,li_07,hall_07}.
In addition, regularization method has also been applied to study the model \citep{crambes_09,yuan_10,cai_12, du_14, wang_15}.
Although the asymptotic properties of estimators of $\beta_{0}$ are
widely discussed in the literature, there is little research on
testing whether $\beta_{0}$ resides in a given finite
dimensional linear subspace, or more specifically, $\beta_{0}\equiv0$.

\begin{figure}
\centering
\includegraphics[width=3.8in]{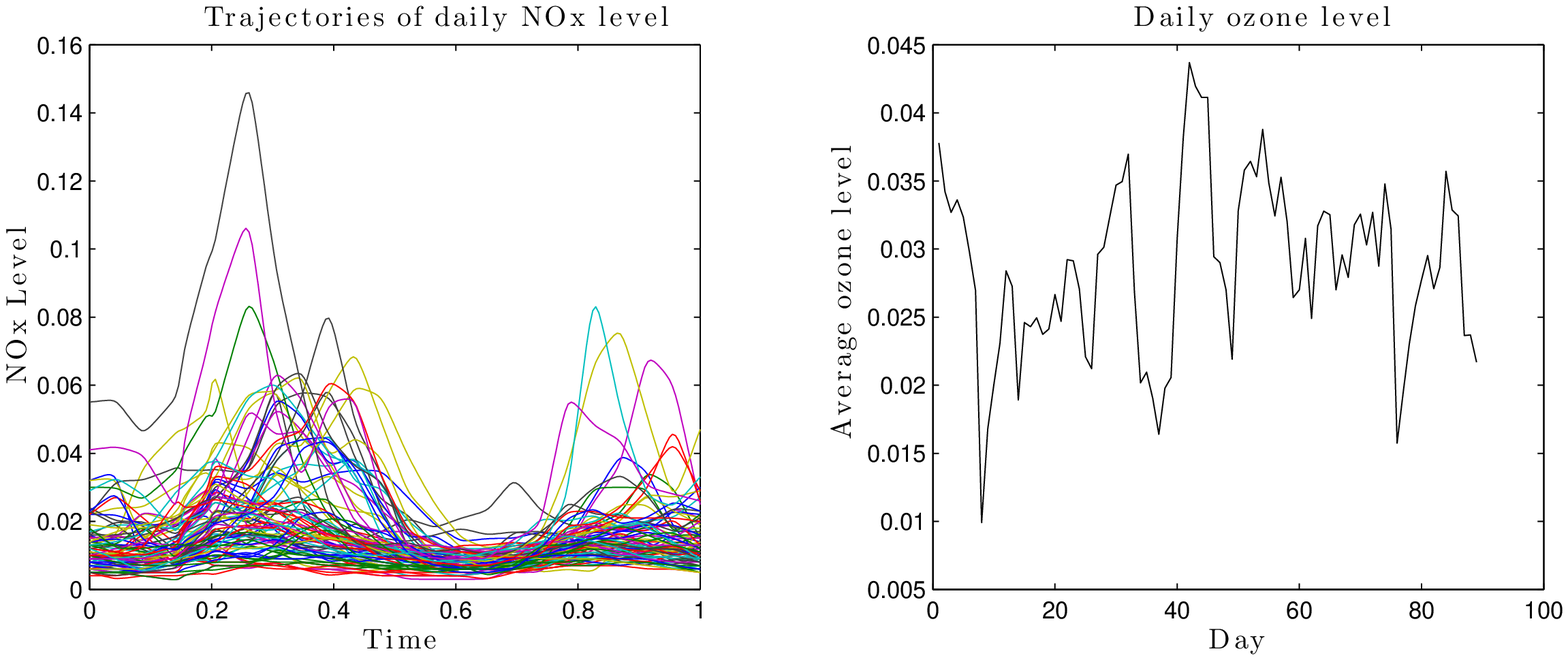}
\includegraphics[width=2.0in]{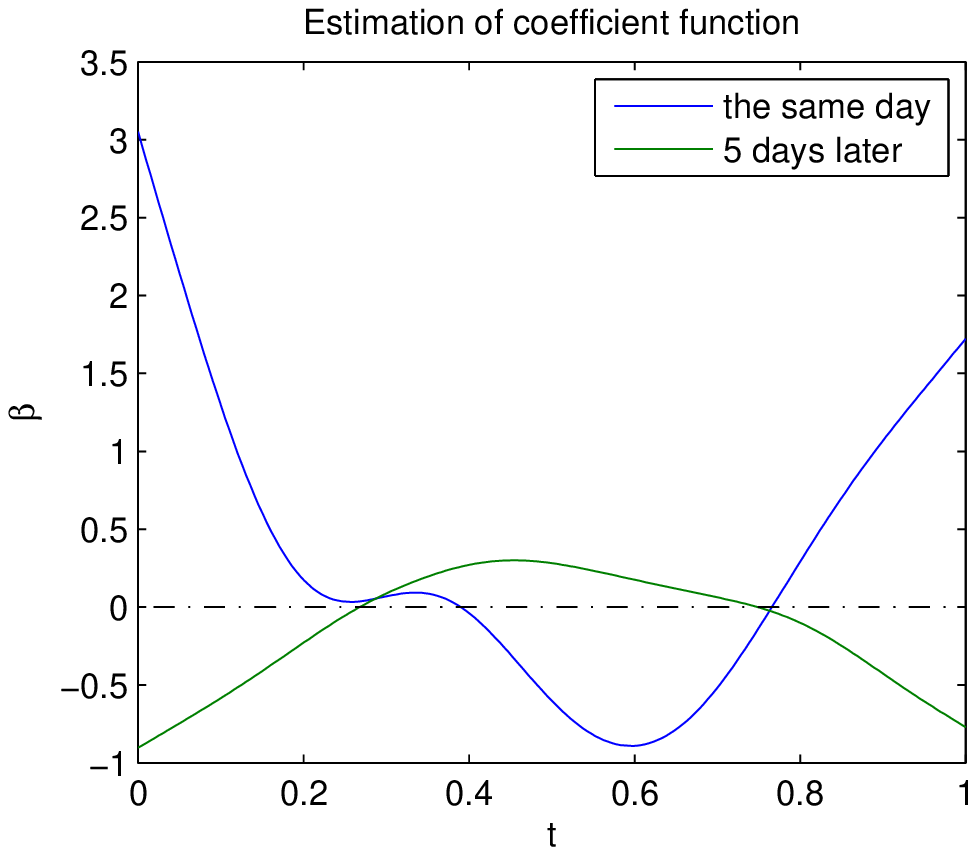}
\caption{Left: the daily trajectories of $\text{NO}_{\text{x}}$ levels. Middle: average $\text{O}_3$ level each day.
Right: estimated coefficient functions}
\label{fig:oz} 
\end{figure}

Take the study of California air quality data as an example.Effects of oxides of nitrogen ($\text{NO}_{\text{x}}$)
on levels of ozone ($\text{O}_3$) is always of great interest to meteorology researchers. The left and middle panels of Figure \ref{fig:oz}
displays the daily trajectories of $\text{NO}_{\text{x}}$ levels as well as the daily average $\text{O}_3$ levels
in the city of Sacramento from June 1 to August 31 in 2005. 
If we take daily $\text{NO}_{\text{x}}$ trajectory as predictor $X(t)$ and average $\text{O}_3$ level as
$Y$, then an absent effect will be indicated by a zero slope function in model (\ref{equ:flr2}). 
The right panel of Figure \ref{fig:oz} plots the estimated slope functions under two settings, 
(1) response $Y$ is the $\text{O}_3$ level of the same day as $\text{NO}_{\text{x}}$ trajectory, and 
(2) response is the $\text{O}_3$ level five days later after the recorded $\text{NO}_{\text{x}}$ trajectory. 
Under setting (1), the estimated slope function has a large magnitude and a clear curve. 
This indicates that the true slope function in this model is very unlikely to be a zero function, that is to say, a day's $\text{NO}_{\text{x}}$ has a strong effect on its $\text{O}_3$ level. 
On the other hand, the estimated slope function under setting (2) stays close to zero and the slight curvature of this estimated slope function may due to randomness of the data, with the true $\beta_{0}$ residing in a zero null space. 
In other words, a day's $\text{NO}_{\text{x}}$ may barely have any effect on the $\text{O}_3$ level five days later.
However to draw a statistical conclusion under a certain significant level on whether there is
still some effect on the $\text{O}_3$ level from the $\text{NO}_{\text{x}}$ level five days ago, we need a well-designed
testing procedure.

\cite{cardot_03} proposed a test statistic based on
the first $k$ functional components of $X$.
However, selection of $k$ is a difficult problem.
Some computational methods have
been proposed to resolve this issue without theoretical guarantee on
the power \citep{cardot_04,gonzalez_11}.
For more recent work, \cite{hilgert_13}
used the functional principle component approach to test the nullity of the slope
function,
and established that their procedures
are minimax adaptive to the unknown regularity of the slope. 
In particular, they assumed that $\beta_0\in {\cal E}_a(L)$ where
\[
{\cal E}_a(L) = \Big\{\beta\in L_2[0, 1]: \sum_{k=1}^\infty a_k^{-2} \big\langle\beta, \varphi_k \big\rangle^2 \le L^2\Big\},
\]
with $\langle \beta, \varphi_k\rangle=\int_0^1 \beta(t)\varphi_k(t)dt$, and $\varphi_{k}$'s are eigenfunctions of the covariance $\Gamma$. 
The smoothness of $\beta_0$ is characterized by the decay rate of $a_k$. 
${\cal E}_a(L)$ is essentially a reproducing kernel Hilbert space (RKHS), 
denoted by ${\cal H}(K)$, with a specific reproducing kernel $K(t, s) = \sum_{k=1}^\infty a_k^2 \varphi_k(t)\varphi_k(s)$. When their underline assumption that, kernel $K$ and $\Gamma$ are well aligned, is not satisfied, their methods may not perform well.
\cite{Lei_14}
developed a method simultaneously testing the slope vectors in a sequence
of functional principal components regression models, and showed that
under certain conditions, his method is uniformly powerful over a
class of smooth alternatives. 
However, the principal-component-based methods
are successful upon the assumption that the slope function $\beta(t)$
can be well represented by the leading functional principal components
of $X$. \cite{cai_12} showed that, for the benchmark Canadian weather data, the estimated Fourier coefficients of the slope function with respect to the eigenfunctions of the sample covariance function do not decay at all, which is a typical example for the case that the slope function is not well represented by the leading principal components. \cite{shang2015} proposed a roughness regularization approach in making non-parametric inference for generalized functional linear model, including a theoretical result on the upper bound. 


In this paper, we propose an adaptive and minimax optimal testing procedures on detecting the nullity
of the slope function in functional linear model
using smoothing splines. 
Let $\Gamma(s,t)$ denote the
covariance function of $X$. $\Gamma$ can also be taken as a nonnegative
definite operator with $\Gamma f=\int_{0}^{1}\Gamma(\cdot,~t)f(t)dt$
for $f\in L_{2}$. 
We wish to test the null hypothesis $H_{0}:\beta\equiv0$
against the composite nonparametric alternative that
$\beta_{0}$ is separated away from zero in terms of a $L_{2}$-norm
induced by the operator $\Gamma$, i.e. $\|\beta_{0}\|_{\Gamma}\ge\varrho_{n}$,
where $\|\beta\|_{\Gamma}^{2}=\langle\Gamma\beta,\beta\rangle$ with
$\langle\beta,\gamma\rangle=\int_{0}^{1}\beta(t)\gamma(t)dt$. Then assuming
that the unknown slope function $\beta_{0}$ possesses some smoothness
properties, therefore, we arrive at the following alternative: 
$H_{1}:{\cal F}_{\Gamma}(\rho_{n})=\{\beta:\|\beta\|_{\Gamma}\ge\rho_{n}\}$.
The radius $\rho_{n}$ characterizes the sensitivity of the
test. We investigate the optimal decay rate of the radius
$\rho_{n}$, under which the test with prescribed probabilities of
errors is still possible.

The paper is organized as follows. In Section \ref{sec:upper}, a
smoothing spline estimate for the slope function is introduced, and
a generalized likelihood ratio test based on this smoothing spline
estimate is proposed. 
In Section \ref{sec:lower},
we show that our test is optimal in the sense that it achieves the minimax lower bound,
which is joint determined by the smoothing spline kernel and the covariance kernel.
Section \ref{sec:numerical} demonstrates the finite sample performance of
the test under different simulated setups. Later in this section come
more details about the air quality example. 

\section{Generalized Likelihood Ratio Test\label{sec:upper}}

\subsection{Notation and definitions}
Since our main focus is on the coefficient function $\beta(t)$, we assume both
$X$ and $Y$ are centered, i.e., $E(Y)=0$ and $E(X(t))=0$ for all $t$. Therefore by taking
expectation over both sides of (\ref{equ:flr}), we have $\alpha_{0}=0$. 
Let $(X_{i},Y_{i}),i=1,\ldots,n$ be independent and identically distributed observations sampled from the model.
Then model (\ref{equ:flr}) can be rewritten as 
\begin{equation}
Y_{i}=\int_{0}^{1}\beta_{0}(t)X_{i}(t)dt+\epsilon_{i},~~~~i=1,\ldots,n.\label{equ:flr2}
\end{equation}

$\beta_{0}(t)$ is considered to reside in the Sobolev space $W_{2}^{m}$ of order $m$, defined as 
\[
W_{2}^{m}=\Big\{\beta:[0,1]\rightarrow\mathbb{R}\Big|\beta,\beta',\ldots,\beta^{(m-1)}\text{ are absolutely continuous and }\beta^{(m)}\in L_{2}[0,1]\Big\}.
\]
Equipting $W_{2}^{m}$ with a reproducing kernel
\[
K(t,s)=\sum_{k=0}^{m-1}\frac{s^{k}t^{k}}{(k!)^{2}}+R(t,s), 
\quad \text{where }
R(t,s)=\int_{0}^{1}\frac{(s-u)_{+}^{m-1}(t-u)_{+}^{m-1}}{\{(m-1)!\}^{2}}du,
\]
it becomes a reproducing kernel Hilbert
space \citep{wahba_90}, denoted as ${\cal H}(K)$.

Let $T_{0}$ and $T_{1}$ be operators on $L_{2}[0,1]$ such that
\[
T_{0}X(t)=\int_{0}^{t}X(s)ds~~~~~\mbox{and}~~~~~T_{1}X(t)=\int_{t}^{1}X(s)ds.
\]
It follows Fubini's theorem that $\langle f,T_{0}g\rangle=\langle T_{1}f,g\rangle$,
and thus $T_{0}$ is the adjoint operator to $T_{1}$. Further, define
that $T_{0}^{k}X(t)=T_{0}T_{0}^{k-1}X(t)$ and $T_{1}^{k}X(t)=T_{1}T_{1}^{k-1}X(t)$
for $k\ge2$. Therefore, $T_{0}^{k}$ is the adjoint operator to $T_{1}^{k}$, and
\begin{align*}
T_{0}^{k}X(t)=\int_{0}^{1}\frac{(t-s)_{+}^{k-1}}{(k-1)!}X(s)ds,~~~~~T_{1}^{k}X(t)=\int_{0}^{1}\frac{(s-t)_{+}^{k-1}}{(k-1)!}X(s)ds.
\end{align*}
In particular, 
\begin{equation*}
R=T_{0}^{m}T_{1}^{m}.\label{equ:RT}
\end{equation*}
Observe that $R$ differs from $K$ only by a polynomial of degree
less than or equal to $m$. Therefore, their eigenvalues
have the same decay rate. 

The following notations will be used in estimating slope function 
and then constructing test statistic. Denote ${\bf X}(t)=(X_{1}(t),\ldots,X_{n}(t))^{T}$
and sample
covariance function $\hat{\Gamma}(t,s)=n^{-1}{\bf X}(t)^{T}{\bf X}(s)$. 
Let $\widetilde{X}(1)\in\mathbb{R}^{m\times n}$
be an $m$ by $n$ matrix with the $(i,j)'s$ element $(\widetilde{X}(1))_{i,j}=T_{0}^{i}X_{j}(1)$
and $\hat{H}=n^{-1}\widetilde{X}(1)\widetilde{X}(1)^{T}$. Define
a matrix $\hat{B}=\frac{1}{n}\widetilde{X}(1)^{T}\hat{H}^{-1}\widetilde{X}(1)$,
then $\hat{B}$ is an $n\times n$ idempotent matrix with $\hat{B}^{2}=\hat{B}$.
Finally, define an operator $\hat{Q}$ as $\hat{Q}(t,s)=n^{-1}\hat{U}(t)^{T}\hat{U}(s)$,
where $\hat{U}(t)$ is a random function vector such that 
\begin{align*}
\hat{U}(t)=(I_{n}-\hat{B})T_{0}^{m}{\bf X}(t).
\end{align*}
It is easy to see that 
\begin{align*}
\hat{Q}=n^{-1}T_{0}^{m}{\bf X}^{T}(I_{n}-\hat{B})T_{0}^{m}{\bf X}=T_{0}^{m}\big(\hat{\Gamma}-\hat{\Gamma}_{0}\big)T_{1}^{m},\label{equ:hatQ}
\end{align*}
where 
\[
\hat{\Gamma}_{0}(t,s)=\frac{1}{n}{\bf X}(t)^{T}B{\bf X}(s),
\]
is a degenerated operator with at
most $m$ eigenvalues. Hence, the eigenvalues of $\hat{Q}$, $T_{0}^{m}\hat{\Gamma}T_{1}^{m}$
and further $T\Gamma T^{*}$ have the same decay rate.

\subsection{The smoothing spline estimator}

In this section, we study the smoothing spline estimate which will
be used to construct the generalized likelihood ratio test in the
next session. Let $\hat{\beta}$ be the smoothing spline estimate
such that $\hat{\beta}\in W_{2}^{m}$ minimizes 
\begin{equation}
\frac{1}{n}\sum_{i=1}^{n}\Big\{Y_{i}-\int_{0}^{1}\beta(t)X_{i}(t)dt\Big\}^{2}+\lambda\int_{0}^{1}\Big\{\beta^{(m)}(s)\Big\}^{2}ds,\label{equ:obj0}
\end{equation}
where $\lambda>0$ is the smoothing parameter. Next theorem provides
the characterization of $\hat{\beta}$.
\begin{theorem}
\label{thm:beta} Denote ${\bf Y}=(Y_{1},\ldots,Y_{n})^{T}$ and operator
$\hat{Q}^{+}=(\lambda I+\hat{Q})^{-1}$. 

(a). The $m$th derivative of $\hat{\beta}$ is 
\begin{equation*}
\hat{\beta}^{(m)}=(-1)^{m}\frac{1}{n}\hat{Q}^{+}\hat{U}^{T}{\bf Y}.\label{equ:betam}
\end{equation*}

(b). Let $\widehat{\Upsilon}(1)=\Big[\hat{\beta}(1),-\hat{\beta}'(1),\ldots,(-1)^{m-1}\hat{\beta}^{(m-1)}(1)\Big]^{T}$.
We have 
\begin{equation*}
\widehat{\Upsilon}(1)=\frac{1}{n}\hat{H}^{-1}\widetilde{X}(1)\Big\{I_{n}-\frac{1}{n}\int_{0}^{1}T_{0}^{m}{\bf X}(s)~\hat{Q}^{+}\hat{U}(s)^{T}ds\Big\}{\bf Y}.\label{equ:beta1}
\end{equation*}
\end{theorem}

Theorem \ref{thm:beta} provides a brand new approach to compute $\hat{\beta}$
explicitly over the infinitely dimensional function space ${\cal H}(K)$.
This observation is important to both numerical implementation and
asymptotic analysis. The explicit formula for $\hat{\beta}$ is 
\begin{equation}
\hat{\beta}(t)=\widehat{\Upsilon}(1)^{T}\zeta(t)+(-1)^{m}\int_{0}^{1}\hat{\beta}^{(m)}(s)\frac{(s-t)_{+}^{m-1}}{(m-1)!}ds=\Pi_{t}{\bf Y}\label{equ:beta_hat}
\end{equation}
where $\zeta(t)=\Big[1,(1-t),\frac{(1-t)^{2}}{2!},\ldots,\frac{(1-t)^{m-1}}{(m-1)!}\Big]^{T}$,
and
\[
\Pi_{t}=\frac{1}{n}\zeta(t)^{T}\hat{H}^{-1}\widetilde{X}(1)\Big\{I_{n}-\frac{1}{n}\int_{0}^{1}T_{0}^{m}{\bf X}(s)~\hat{Q}^{+}\hat{U}(s)^{T}ds\Big\}+\frac{1}{n}T_{1}^{m}\hat{Q}^{+}\hat{U}(t)^{T}.
\]
Therefore, $\hat{\beta}$ is a linear function of the response ${\bf Y}$
with $\Pi_{t}$ as the hat matrix.

\subsection{Generalized likelihood ratio test\label{sec:upper2}}

Assuming that $\epsilon_{i}$
follows normal distribution, the conditional log-likelihood
function for (\ref{equ:flr2}) becomes 
\[
\ell_{n}(\beta,\sigma)=-n\log(\sqrt{2\pi}~\sigma)-\frac{1}{2\sigma^{2}}\sum_{i=1}^{n}\Big(Y_{i}-\int\beta X_{i}\Big)^{2}.
\]
Define the residual sum of squares under the null and alternative
hypothesis as follows: 
\[
\mbox{RSS}_{0}=\sum_{i=1}^{n}Y_{i}^{2},~~~\mbox{RSS}_{1}=\sum_{i=1}^{n}(Y_{i}-\int\hat{\beta}X_{i})^{2}.
\]
Then the logarithm of the conditional maximum likelihood ratio test
statistic is given by 
\begin{equation}
\tau_{n,\lambda}=\ell_{n}(\hat{\beta},\hat{\sigma}_{1})-\ell_{n}(0,\hat{\sigma}_{0})=\frac{n}{2}\log\frac{\mbox{RSS}_{0}}{\mbox{RSS}_{1}},\label{equ:tau}
\end{equation}
where $\hat{\sigma}_{1}^{2}=\mbox{RSS}_{1}/n$ and $\hat{\sigma}_{0}^{2}=\mbox{RSS}_{0}/n$.
Define an $n\times n$ matrix $A_{n}=A_{n}(\boldsymbol{X})$ as

\[
A_{n}=\frac{1}{n}\int_{0}^{1}\hat{U}(t)\hat{Q}^{+}\hat{U}(t)^{T}dt-\frac{1}{2n}\int_{0}^{1}\int_{0}^{1}\hat{Q}^{+}\hat{U}(t)\hat{Q}(t,s)\hat{Q}^{+}\hat{U}(s)^{T}dtds+\frac{1}{2}\hat{B}.
\]
Next theorem shows the properties of the test statistic $\tau_{n,\lambda}$.
\begin{theorem}
\label{thm:test}. If $tr(A_{n})=o_{p}(n)$, we have the following
results,

(a). Under $H_{0}:\beta\equiv0$, the likelihood ratio test statistic
$\tau_{n,\lambda}$ is of the form 
\begin{align*}
\tau_{n,\lambda}=z^{T}A_{n}z+o_{p}\big(1\big),
\end{align*}
where $z=\epsilon/\sigma$. Furthermore,let $\mu_{n}=tr(A_{n})$ and
$\sigma_{n}^{2}=2\,tr(A_{n}^{2})$. If $\epsilon_{i}$, $i=1,...,n$
are independent and identically distributed following $\mathcal{N}(0,\sigma^{2})$, then $(\tau_{n,\lambda}-\mu_{n})/\sigma_{n}$has
an asymptotic standard normal distribution.

(b). Under $H'_{1}:{\cal F}'_{K,\Gamma}(\rho_{n})=\Big\{\beta\in{\cal H}(K):\|\beta\|_{\Gamma}=\rho_{n}\Big\},$
if $\rho_{n}^{2}=o(n^{-1/2})$ and $\lambda=o(n^{-1/2})$, then 
\[
\tau_{n,\lambda}=z^{T}A_{n}z+\frac{n}{2\sigma^{2}}\|\beta_{0}\|_{\hat{\Gamma}}^{2}+O_{p}\Big(n\lambda+n^{1/2}\lambda^{1/2}+n^{1/2}\|\beta_{0}\|_{\hat{\Gamma}}\Big).
\]
\end{theorem}

The condition that $tr(A_{n})=o_{p}(n)$ in Theorem \ref{thm:test}
can be satisfied in many cases. In fact, $tr(A_{n})$ can be computed
explicitly. Consider the spectral decomposition of operator $\hat{Q}$,
$\hat{Q}(t,s)=\sum_{j=1}^{\infty}\hat{\kappa}_{j}\hat{\phi}_{j}(t)\hat{\phi}_{j}(s)$,
where $(\hat{\kappa}_{j},\hat{\phi}_{j})$ are (eigenvalue, eigenfunction)
pairs, ordered such that $\hat{\kappa}_{1}\ge\hat{\kappa}_{2}\ge\cdots\ge0$.
We may write $\hat{U}_{X_{i}}(t)=\sum_{k=1}^{\infty}\hat{\xi}_{ik}\hat{\phi}_{k}(t)$.
Since $\hat{Q}(t,s)=n^{-1}\sum_{i=1}^{n}\hat{U}_{X_{i}}(t)\hat{U}_{X_{i}}(s)$,
we have $n^{-1}\sum_{i=1}^{n}\hat{\xi}_{ik}^{2}=\hat{\kappa}_{k}$
and $n^{-1}\sum_{i=1}^{n}\hat{\xi}_{ik}\hat{\xi}_{ij}=0$ for $k\neq j$.
It is not hard to obtain that 
\[
tr(A_{n})=\sum_{k=1}^{\infty}\frac{\hat{\kappa}_{k}(\lambda+\frac{1}{2}\hat{\kappa}_{k})}{(\lambda+\hat{\kappa}_{k})^{2}}+\frac{m}{2}.
\]
Furthermore, Lemma~3  shows that $tr(A_{n})=O_{p}(\sum_{k=1}^{\infty}\frac{s_{k}}{\lambda+s_{k}})$,
which is determined by the order of $\lambda$ and the decay rate
of $s_{k}$, the sorted eigenvalues of linear operator $T\Gamma T^{*}$.
More specifically, if $s_{k}$ has a polynomial decay rate as $s_{k}\asymp k^{-2r}$,
for some $r>1/2$, then $tr(A_{n})=O_{p}(\lambda^{-1/2r})$, while
if $s_{k}$ has an exponential decay rate as $s_{k}\asymp e^{-2rk}$
for some $r>0$, then $tr(A_{n})=O(\log\mbox{\ensuremath{\lambda}}^{-1})$.
In both cases, $tr(A_{n})=o_{p}(n)$ will be satisfied once we choose
a proper $\lambda$. The optimal order of $\lambda$ will be shown
later in Theorem \ref{thm:ttheory}, followed by a data-driven procedure of
choosing $\lambda$. 

Based on Theorem $\ref{thm:test}$, we have an $\alpha$ level testing procedure that,
we reject $H_{0}$ when $\frac{\tau_{n,\lambda}-\mu_{n}}{\sigma_{n}}>z_{\alpha}$
where $z_{\alpha}$ is the upper $\alpha$ quantile of the standard
normal distribution. In the next section, we will show that the power
function of this test is asymptotically one at the minmax optimal
rate.

\section{Optimal Test}
\label{sec:lower} 

\subsection{Minimax lower bound}

Let $\phi_{n}$ be a measurable function of the observations taking
values at two points $\{0,1\}$. We accept $H_{0}$ if $\phi_{n}=0$,
and reject $H_{0}$ if $\phi_{n}=1$. The probability of type I error,
denoted by $\alpha_{0}(\phi_{n})$, is 
\[
\alpha_{0}(\phi_{n})=\mathbb{P}_{0}(\phi_{n}=1),
\]
where $\mathbb{P}_{0}$ is the probability measure on the space of
observations corresponding to $H_{0}$. The probability of type II
error, denoted by $\alpha_{1}(\phi_{n})$, is 
\[
\alpha_{1}(\phi_{n},\rho_{n})=\sup_{\beta\in{\cal F}_{K,\Gamma}(\rho_{n})}\mathbb{P}_{\beta}(\phi_{n}=0),
\]
where $\mathbb{P}_{\beta}$ is the probability measure corresponding
to a particular slope function $\beta$. Let 
\[
\gamma_{n}(\phi_{n},\rho_{n})=\alpha_{0}(\phi_{n})+\alpha_{1}(\phi_{n},\rho_{n}),
\]
which measures the error of the test $\phi_{n}$ by summarizing
probability of the type I and type II errors. Fix a number
$0<\gamma<1$. A sequence $\rho_{n}\rightarrow0$
as $n\rightarrow\infty$ is called the minimax rate of testing if: 
\begin{description}
\item [{(i)}] For any sequence $\rho'_{n}$ such that $\rho'_{n}/\rho_{n}\rightarrow0$,
we have $\lim\inf_{n\rightarrow\infty}\inf_{\phi_{n}}\gamma_{n}(\phi_{n},\rho_{n}')\ge\gamma$; 
\item [{(ii)}] There exists a test $\phi_{n}^{*}$ such that $\lim\sup_{n\rightarrow\infty}\gamma_{n}(\phi_{n}^{*},\rho_{n})\le\gamma$. 
\end{description}
For the given reproducing kernel $K$, let $T$ and $T^{*}$ be two
operators acting on $L_{2}[0,1]$ such that $K=TT^{*}$, where $T^{*}$
is the adjoint operator to $T$ with $\langle f,Tg\rangle=\langle T^{*}f,g\rangle$.
Consider the linear operator $T\Gamma T^{*}$. It follows from the spectral
theorem that 
\[
T\Gamma T^{*}(t,s)=\sum_{k=1}^{\infty}s_{k}\varphi_{k}(t)\varphi_{k}(s),
\]
where $s_{1}\ge s_{2}\ge\cdots>0$ are the eigenvalues of the operator
$T\Gamma T^{*}$ and $\varphi_{k}$'s are the corresponding eigenfunctions.
For any two sequences $a_{k},b_{k}>0$, $a_{k}\asymp b_{k}$ means
that $a_{k}/b_{k}$ is bounded away from zero and infinity as $k\rightarrow\infty$.
\begin{theorem}
\label{thm:lower} Assume $\epsilon_{i}$, $i=1,...,n$ are independent and identically distributed following $\mathcal{N}(0,\sigma^{2})$.
Let $\{s_{k}:k\ge1\}$ be the sorted eigenvalues of the linear operator
$T\Gamma T^{*}$. 

(a). When $s_{k}\asymp k^{-2r}$ for some constant $r>1/2$, let 
\begin{equation}
\rho_{n}=n^{-2r/(1+4r)}.\label{equ:rho}
\end{equation}
If $\rho_{n}'$ is such that $\rho_{n}'/\rho_{n}\rightarrow0$ as
$n\rightarrow\infty$, then 
\[
\liminf_{n\rightarrow\infty}\inf_{\phi_{n}}\gamma_{n}(\phi_{n},\rho_{n}')\ge1.
\]

(b). When $s_{k}\asymp e^{-2rk}$ for some constant $r>0$, let 
\begin{equation}
\rho_{n}=\Big(\frac{\log n}{2rn^{2}}\Big)^{1/4}.\label{equ:rho2}
\end{equation}
If $\rho_{n}'$ is such that $\rho_{n}'/\rho_{n}\rightarrow0$ as
$n\rightarrow\infty$, then 
\[
\liminf_{n\rightarrow\infty}\inf_{\phi_{n}}\gamma_{n}(\phi_{n},\rho_{n}')\ge1.
\]

\end{theorem}

The cholesky decomposition of the operator $K=TT^{*}$ is not unique, 
and $T$ is not necessarily a symmetric operator. If we would like
$T$ to be a symmetric operator, we may choose $T=T^{*}=K^{1/2}$.
It is shown in the next proposition that the decay rate of
the eigenvalues of the operator $T\Gamma T^{*}$ and $K^{1/2}\Gamma K^{1/2}$ have the same asymptotic
order.
\begin{proposition}
\label{prop:kt} Let $K=TT^{*}$, where $T^{*}$ is adjoint to $T$.
The eigenvalues of the two operators $T\Gamma T^{*}$ and $K^{1/2}\Gamma K^{1/2}$
have the same decay rate. 
\end{proposition}
The minimax lower bound for the excess prediction risk has been established
by \cite{cai_12}. Suppose the $kth$ eigenvalues of the linear operator
$K^{1/2}\Gamma K^{1/2}$ is of order $k^{-2r}$ for some constant
$0<r<\infty$, then 
\[
\lim_{a\rightarrow0}\lim_{n\rightarrow\infty}\inf_{\hat{\beta}}\sup_{\beta_{0}\in H(K)}\mathbb{P}\Big(\big\|\hat{\beta}-\beta_{0}\big\|_{\Gamma}\ge an^{-\frac{r}{2r+1}}\Big)=1.
\]
It turns out that the optimal separating rate $\rho_{n}$ for testing
differs from the optimal rate for the problem of prediction. Similar
situation arises in the setting of nonparametric regression. 

Consider a special case that the reproducing kernel $K$ is perfectly aligned
with $\Gamma$, i.e., $K(s,t)=\sum_{k=1}^{\infty}a_{k}^{2}\psi_{k}(t)\psi_{k}(s)$
and $\Gamma(t,s)=\sum_{k=1}^{\infty}\eta_{k}\psi_{k}(t)\psi_{k}(s)$.
In this case, it is easy to see that $K^{1/2}\Gamma K^{1/2}(t,s)=\sum_{k=1}^{\infty}\eta_{k}a_{k}^{2}\psi_{k}(t)\psi_{k}(s)$,
which indicates that $s_{k}=\eta_{k}a_{k}^{2}$. This special case
has been studied in \cite{hilgert_13}.

\subsection{Optimal adaptive test}

Now back to the generalized likelihood ratio test. Recall that the
test statistic $\tau_{n,\lambda}$ has an asymptotic normal distribution
with mean $\mu_{n}=tr(A_{n})$ and variance $\sigma_{n}^{2}=2\,tr(A_{n}^{2})$.
Concerning the distribution of the random function $X$, we shall
assume that 
\begin{description}
\item [{(A1).}] $X$ has a finite fourth moment, i.e., $\int_{0}^{1}E(X^{4})<\infty$
and 
\[
E\Big(\big\langle X,\psi_{k}\big\rangle^{4}\Big)\le C\Big(E\big\langle X,\psi_{k}\big\rangle^{2}\Big)^{2}~~~~\mbox{for }k\ge1,
\]
where $C>0$ is a constant and $\psi_{k}$'s are eigenfunctions of
$\Gamma$. \end{description}
\begin{theorem}
\label{thm:ttheory} Assume (A1) holds and $\epsilon_{i}$, $i=1,...,n$
are independent and identically distributed following $\mathcal{N}(0,\sigma^{2})$. Let $\{s_{k}:k\ge1\}$ be the
sorted eigenvalues of the linear operator $T\Gamma T^{*}$. 

(a). When $s_{k}\asymp k^{-2r}$ for some constant $r>1/2$. Choose
$
\lambda=cn^{-4r/(4r+1)}$,
for some $c>0$. Then $\mu_{n}$ and $\sigma_{n}^{2}$ are of order
$O_{p}(n^{2/(4r+1)})$, and for any sequence $c_{n}\rightarrow\infty$,
the power function of the generalized likelihood ratio test is asymptotically
one: 
\[
\inf_{\beta\in{\cal F}_{K,\Gamma}(c_{n}\cdot\rho_{n}):\|\beta\|_{\Gamma}\ge c_{n}n^{-2r/(4r+1)}}\mathbb{P}_{\beta}\Big(\frac{\tau_{n,\lambda}-\mu_{n}}{\sigma_{n}}>z_{\alpha}\Big)\longrightarrow1,
\]
where $z_{\alpha}$ is the upper $\alpha$ quantile of the standard
normal distribution and $\rho_{n}$ is given in (\ref{equ:rho}).

(b). Assume $s_{k}\asymp\exp(-2rk)$ for some constant $r>0$. Choose
$\lambda$ such that 
\[
\log\lambda^{-1}=O(\log n),\quad\quad\lambda^{-1}n^{-1}=O(1),\quad\quad and\quad\lambda=o(n^{-1/2}).
\]
 Then $\mu_{n}$ and $\sigma_{n}^{2}$ are of order $O_{p}\{\log n/(2r)\}$,
and for any sequence $\tilde{c}_{n}\rightarrow\infty$, 
\[
\inf_{\beta\in{\cal F}_{K,\Gamma}:\|\beta\|_{\Gamma}\ge\tilde{c}_{n}\{\log n/(2rn^{2})\}^{1/4}}\mathbb{P}_{\beta}\Big(\frac{\tau_{n,\lambda}-\mu_{n}}{\sigma_{n}}\ge z_{\alpha}\Big)\longrightarrow1.
\]
\end{theorem}

The optimal smoothing parameters for prediction and testing are different.
When $\kappa_{k}\asymp k^{-2r}$, if we choose $\lambda=\tilde{\lambda}$
to be of order $n^{-2r/(2r+1)}$, which is the optimal order for prediction,
the rate of the testing will be slower than the optimal rate given
in Theorem \ref{thm:lower}. Specifically, there exists a $\beta\in{\cal F}_{K,\Gamma}$
satisfying $\|\beta\|_{\Gamma}=n^{-(r+d)/(2r+1)}$ with $d>1/8$ such
that the power function of the test at the point $\beta$ is bounded
by $\alpha$, namely 
\[
\limsup_{n\rightarrow\infty}\mathbb{P}_{\beta}\Big(\tau_{n,\tilde{\lambda}}>\mu_{n}+z_{\alpha}\sigma_{n}\Big)\le\alpha.
\]
As we see in part (b), when $s_{k}$ is exponentially decayed, the
choice of $\lambda$ is more flexible. For example, any $n^{d}$ for
$-1\leq d<-\frac{1}{2}$ , could guarantee an optimal test. 

Considering $\lambda^{*}$ such that 
\[
\lambda^{*}=\arg\min_{\lambda\ge0}\Big(\lambda+\frac{1}{n}\sum_{k=1}^{\infty}\frac{\kappa_{k}}{\sqrt{\lambda}+\kappa_{k}}\Big),
\]
where $\kappa_{k}$'s are eigenvalues of $Q=T_{0}^{m}\Gamma T_{1}^{m}$.
$\lambda^{*}$ is well-defined, since 
\begin{align*}
\sum_{k=1}^{\infty}\kappa_{k} & =\int_{0}^{1}Q(t,t)dt=E\big\langle T_{0}^{m}X,T_{1}^{m}X\big\rangle\le C_{1}\int_{0}^{1}E(X^{2})<\infty.
\end{align*}
It is not hard to see that $\lambda^{*}\asymp n^{-4r/(4r+1)}$ if
$\kappa_{k}\asymp k^{-2r}$, while $\lambda^{*}\asymp n^{-1}$ if
$\kappa_{k}\asymp e^{-2rk}$. Therefore an estimated $\lambda^{*}$
can be used as our choice of the smoothing parameter. It is natural
to use $\tilde{Q}=T_{0}^{m}\hat{\Gamma}T_{1}^{m}$ as an estimate
of $Q$. The following Theorem gives an adaptive estimation of $\lambda$.
\begin{theorem}
\label{thm:adaptive} Assume (A1) holds. Denote by $\tilde{\kappa}_{1}\ge\tilde{\kappa}_{2}\ge\cdots\ge0$
the eigenvalues of $\tilde{Q}$. Choosing $\tilde{\lambda}$ as 
\begin{equation}
\tilde{\lambda}=\arg\min_{\lambda\ge0}\Big(\lambda+\frac{1}{n}\sum_{k=1}^{\infty}\frac{\tilde{\kappa}_{k}}{\sqrt{\lambda}+\tilde{\kappa}_{k}}\Big).\label{eq:lambda tilde}
\end{equation}
When $s_{k}\asymp k^{-2r}$ for some constant $r>1/2$, there exist
constants $0<c_{1}<c_{2}<\infty$ such that 
\[
\lim_{n\rightarrow\infty}\mathbb{P}\Big(c_{1}<\frac{\tilde{\lambda}}{\lambda_{o}}<c_{2}\Big)=1
\]
where $\lambda_{o}=cn^{-4r/(4r+1)}$ for some $c>0$.
\end{theorem}
Theorem $\ref{thm:adaptive}$ verifies that $\tilde{\lambda}$ chosen
by ($\ref{eq:lambda tilde}$) is of the proper order. Simulations
also show that as long as $X(s)$ and $Y$ are at a proper scale, say
ranging at the level of $[-10,10]$, we can directly use the $\tilde{\lambda}$
without worrying about multiplying a constant. However we need to
be more careful when $X$ and $Y$ are numerically at a different
scale. As for the case when $\kappa_{k}$ is exponentially decayed,
the proper $\lambda$ has a much larger range. We can still use ($\ref{eq:lambda tilde}$)
to get a proper $\lambda$.

\section{Numerical Studies \label{sec:numerical}}

\subsection{Simulation\label{sec:Simulation}}

Consider the case that slope function $\beta(t)$ is in the Soblev
space $W_{2}^{2}$. The penalty function in (\ref{equ:obj0}) becomes
$\lambda\int_{0}^{1}\beta^{''}(s)^{2}ds.$ Following a similar setup
as that in Yuan and Cai (2010), we generate the covariate function
$X\text{(t)}$ by:
$X(t)=\sum_{k=1}^{50}\zeta_{k}Z_{k}\phi_{k}(t)$.
where $Z_{k}$'s are independently sampled from $Unif[-\sqrt{3},\sqrt{3}]$
and $\phi_{k}$'s are Fourier basis with $\phi_{1}=1$ and $\phi_{k+1}(t)=\sqrt{2}\cos(k\pi t)$
for $k\geq1$.  We have two settings for $\zeta_{k}$. For setup
1, let $\zeta_{k}=(-1)^{k+1}k^{-v/2}/||\zeta||$,
where $\zeta=(\zeta_{1},...,\zeta_{50})^T$ and $||\cdot||$ indicates
$\mathcal{L}_{2}$ norm. The normalizing term $||\zeta||^{-1}$ is
added to rule out the potential effect from the magnitude of $X(s)$.
For setup 2, $\zeta$ is chosen as 
\begin{equation*}
\zeta_{k}=\begin{cases}
\begin{array}{c}
1\\
0.2(-1)^{k+1}(1-0.0001k)\\
0.2(-1)^{k+1}[5(k/5)]^{-v/5}-0.0001(k\,mod\,5)
\end{array} & \begin{array}{c}
k=1\\
2\leq k\leq4\\
k\geq5
\end{array}\end{cases}.
\end{equation*}
The eigenvalues of the covariance function of $X(t)$
are $\zeta_{k}^{2}$'s, the decay rate of which is determined by $\nu$.
In both cases, let $\nu=1.1,\,1.5,\,2,\,4.$ With the same basis,
the true slope function $\beta_{0}$ is generated as: 
$
\beta_{0}=B\cdot\sum_{i=1}^{50}(-1)^{k+1}k^{-2}\phi_{k},
$
where B is a constant to control the norm of $\beta_{0}$. For both
setups, a set of B ranging from 0 to 1 is examined.
Response $Y$ is generated through the functional regression model
with $\varepsilon\sim N(0,\,1)$. Sample
size $n=50,\,100,\,200$ are adopted to appreciate the effect of sample
size. 

For each simulated dataset, smoothing parameter $\lambda$ is chosen
based on ($\ref{eq:lambda tilde}$), $\hat{\beta}(t)$ is estimated
by (\ref{equ:beta_hat}), and the testing statistic $\tau_{n,\lambda}$
is calculated as shown in ($\ref{equ:tau}$). According to Theorem
\ref{thm:test}, we reject $H_{0}$ if $|\frac{\tau_{n,\lambda}-\mu_{n}}{\sigma_{n}}|>z_{\alpha/2}$,
with $\alpha=0.05$. To estimate the size and power of our testing
procedure, each setting is repeated 1000 times to get the percentage
of rejecting $H_{0}$.

\begin{table}
\begin{centering}
\begin{tabular}{c|ccc}
\hline 
 & n=50  & n=100  & n=200 \tabularnewline
\hline 
$\nu$=1.1  & 0.066 & 0.058 & 0.059\tabularnewline
\hline 
$\nu$=1.5  & 0.048 & 0.055 & 0.041\tabularnewline
\hline 
$\nu$=2  & 0.051 & 0.055 & 0.045\tabularnewline
\hline 
$\nu$=4  & 0.067 & 0.053 & 0.041\tabularnewline
\hline 
\end{tabular}
\quad\quad
\begin{tabular}{c|ccc}
\hline 
 & n=50  & n=100  & n=200 \tabularnewline
\hline 
$\nu$=1.1  & 0.065 & 0.052 & 0.054\tabularnewline
\hline 
$\nu$=1.5  & 0.063 & 0.046 & 0.057\tabularnewline
\hline 
$\nu$=2  & 0.059 & 0.051 & 0.051\tabularnewline
\hline 
$\nu$=4  & 0.065 & 0.044 & 0.050\tabularnewline
\hline 
\end{tabular}
\par\end{centering}
\caption{Size of the test. Left: setup 1. Right: setup2.}
\label{tab:setup1 size}

\end{table}

Table \ref{tab:setup1 size} shows the size of the test
under different decay rate $\nu$ and sample size $n$ for both setups.
The size of the test stays closer around 0.05. 

Under alternative hypothesis $H_{1}:\beta_{0}\in F_{K,\Gamma}(\rho_{n})$,
the power function of test under different decay rate $\nu$ and sample
size $n$ are shown in Figure \ref{fig:setup1 pf}. It is very clear
that as $B$ increases, $||\beta_{0}||_{\Gamma}$ increases, and therefore
the power of the test increases to 1. Also as expected, under the same
setting, when sample size $n$ goes up, the power should increase,
which 
manifests a steeper slope of the power function in the figure.
What is more interesting in the figure, is how the power is affected
by the decay rate of the eigenvalues of $T_{0}^{m}\Gamma T_{1}^{m}$,
which in our setting is determined by $\nu$. As shown in the figure,
power function with $\nu=4$ always lies on top while that with $\nu=1.1$
always stays the lowest, which perfectly matches Theorem \ref{thm:lower} that the
larger the $\nu$, the faster the decay rate, and therefore the more
powerful the test.

\begin{figure}[h!]
\center
\includegraphics[width=5.6in]{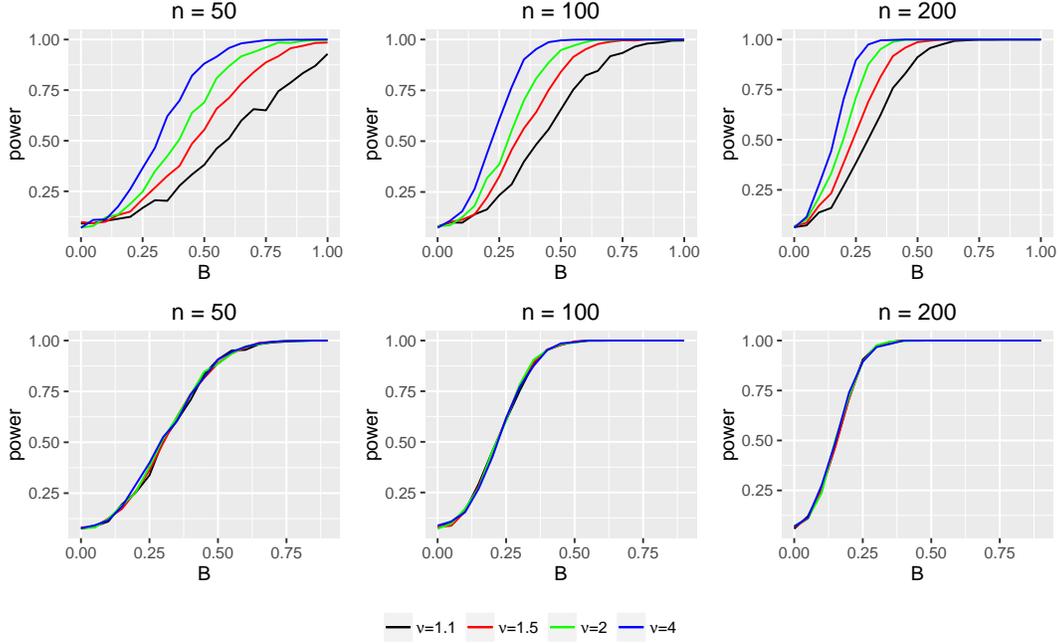}
\caption{Power function of the test for n=50, 100, 200. First row: under setup 1; second row: under setup 2}
\label{fig:setup1 pf} 
\end{figure}

Similarly for setup 2,
the power of the test goes up when
sample size $n$ and $||\beta_{0}||_{\Gamma}$ increase. However the effect
of the decay rate $\nu$ can be hardly seen this time. The reason is that
when choosing $\zeta$ we did not normalize it as we did in setup1.
Therefore even though a larger $\nu$ could lead to a more powerful test,
the magnitude of $X(s)$ is significantly decreased due to the faster
decay rate, and this counter balanced the effect of $\nu$. 


%
%

\subsection{California air quality data}

Back to the California air quality example, as mentioned in the introduction,
we are interest in testing the effect of trajectories of oxides of
nitrogen ($\text{NO}_{\text{x}}$) on the level of ground-level concentrations of ozone
($\text{O}_3$). Data we are using is from the database of California Air Quality
Data. $\text{NO}_{\text{x}}$ levels and $\text{O}_3$ levels of city Sacramento
are recorded from June 1 to August 31 in 2015. There are 91 days on the
record, and 3 days are removed due to severe missing data. For the
rest 89 days, levels of $\text{NO}_{\text{x}}$ are observed at each hour except for 4am
and average $\text{O}_3$ level can also be obtained through the recorded data.
The left panel of Figure \ref{fig:oz} displays the daily trajectories
of $\text{NO}_{\text{x}}$ levels, and the middle panel shows the average $\text{O}_3$ level each
day during the same time period. When applying the proposed testing procedure, every record is rescaled by multiplying
100 due to the small magnitude.

Let $X_{i}(s),\,i=1,...,89$ denote the daily trajectories of $\text{NO}_{\text{x}}$
levels after pre-smoothing and centering, and rescale $s$ so that
$s\in[0,1]$. In the introduction, two types of response variables
are considered, the average $\text{O}_3$ level of the same day as the $\text{NO}_{\text{x}}$
level, and the average $\text{O}_3$ level five days later after the recorded
$\text{NO}_{\text{x}}$ trajectory. More generally we can examine the relation between
the $\text{O}_3$ level of a certain day and the $\text{NO}_{\text{x}}$ level $d$ days before that day. If we
take $Y_{i},\ i=1,...,89$ as the corresponding $\text{O}_3$ level of the day
when $X_{i}$ is recorded. Then the regression function is written
as 
\[
Y_{i+d}=\int_{0}^{1}X_{i}(s)\beta(s)ds+\epsilon_{i},
\]
for a fixed $d$. 

We go through the proposed testing procedure for $d=0,1,...,5$ and
all the p-value are listed in Table \ref{tab: p-value}. We can
see that for $d$ up to 4, the test returns a significant result at
level $\alpha=0.05$, which indicates that daily $\text{NO}_{\text{x}}$ level is
significantly related to the $\text{O}_3$ level up to four days later. It is also
interesting to see that the smallest p-value occurs at $d=1$. A possible
way to interpret it is that instead of the current $\text{NO}_{\text{x}}$ level, the average $\text{O}_3$ level
depends more on the $\text{NO}_{\text{x}}$ level the day before. That is to say
there is a delayed effect of $\text{NO}_{\text{x}}$ level on $\text{O}_3$ level. 

\begin{table}
\caption{P-value}
\centering
\begin{tabular}{c|ccccccc}
\hline \hline
d & 0 & 1 & 2 & 3 & 4 & 5 & 6\tabularnewline
\hline 
p-value & 3.07e-5 & 6.78e-9 & 2.30e-5 & 3.13e-4 & 0.0031 & 0.36 & 0.70\tabularnewline
\hline \hline
\end{tabular}

\label{tab: p-value}
\end{table}

\section{Discussion}

We have so far focused on the case with continuously observed functional predictors. If we have densely observed functional predictors, our framework can be applied similarly. An interesting extension of the current work would be to study the case when having sparsely observed functional predictors with/without measurement error. The ideas of \cite{yao_05b} can be applied. A common strategy is to first have a pre-smoothing step and then apply our methodology. How the number of sparse observations affects the power of the test is beyond the scope of this paper and will be explored in future works. 

A continuation of this paper is to study the optimal testing for the generalized functional linear model with a scalar response and
a functional predictor \citep{du_14}. 
Given the functional predictor, the response is assumed to follow some distribution from the exponential family.  The main difficulty is that the characterization conditions of the slope estimator becomes complex and nontrivial. This problem hinders further studies in the asymptotic properties. We conjecture that the generalized likelihood ratio test will achieve the optimal rate of testing and the optimal rate still depends on the decay rate of $K^{1/2}\Gamma K^{1/2}$. This issue will be addressed in detail in the future.

%

\section{Proofs of Theorems\label{sec:Proofs-of-Theorems}}


\subsection{Proof of Theorem \ref{thm:beta}}

We prove this theorem using the calculus of variation. Denote 
\[
L(\beta)=\frac{1}{n}\sum_{i=1}^{n}\Big\{ Y_{i}-\int_{0}^{1}X_{i}(s)\beta(s)ds\Big\}^{2}+\lambda\int_{0}^{1}\Big\{\beta^{(m)}(s)\Big\}^{2}ds.
\]
For any $\beta,\beta_{1}\in W_{2}^{m}$ and $\delta\in\mathbb{R}$,
\begin{equation}
L(\beta+\delta\beta_{1})-L(\beta)=2\delta L_{1}(\beta,\beta_{1})+O(\delta^{2}),\label{equ:temp1}
\end{equation}
where 
\begin{eqnarray}
L_{1}(\beta,\beta_{1}) & = & -\frac{1}{n}\sum_{i=1}^{n}\big\{ Y_{i}-\int_{0}^{1}X_{i}(s)\beta(s)ds\big\}\big\{\int_{0}^{1}X_{i}(s)\beta_{1}(s)ds\big\}\label{equ:psi1}\\
 &  & ~~~~~~~~~~~+\lambda\int_{0}^{1}\beta^{(m)}(s)\beta_{1}^{(m)}(s)ds.\nonumber 
\end{eqnarray}

By Lemma~1,
if $L_{1}(\beta,\beta_{1})=0$ for all $\beta_{1}\in W_{2}^{m}$,
letting $\mathcal{I}_{1}=\{t\in[0,1]:L_{2}(\beta)\neq0\}$ and $\beta_{1}^{(m)}(t)=-I_{\mathcal{I}_{1}}(t)$
gives 
\[
L_{1}(\beta,\beta_{1})=\int_{\mathcal{I}_{1}}L_{2}(\beta)dt\neq0,
\]
unless $\mathcal{I}_{1}$ is of measure zero. This shows $L_{2}(\beta)=0$
a.e.. This complete the proof of the first part of the theorem.

If $\hat{\beta}$ is the optimal solution, we have 
\[
\hat{\beta}^{(m)}=\frac{(-1)^{m}}{n}\hat{Q}^{+}\hat{U}^{T}{\bf Y}.
\]
It follows from (\ref{equ:beta11}) that 
\[
\hat{H}\widehat{\Upsilon}(1)+\frac{(-1)^{m}}{n}\widetilde{X}(1)\int_{0}^{1}T_{0}^{m}{\bf X}(s)\hat{\beta}^{(m)}(s)ds=\frac{1}{n}\widetilde{X}(1){\bf Y}.
\]
Therefore, the second part of the theorem follows from these two facts.


\subsection{Proof of Theorem \ref{thm:test}}

For part (a), under $H_{0}$ with $\beta_{0}\equiv0$, we have 
\begin{align*}
\frac{1}{n}\mbox{RSS}_{0} & =\frac{1}{n}\epsilon^{T}\epsilon,\\
\frac{1}{n}\mbox{RSS}_{1} & =\frac{1}{n}\epsilon^{T}\epsilon+\|\hat{\beta}-\beta_{0}\|_{\hat{\Gamma}}^{2}-\frac{2}{n}\epsilon^{T}\int(\hat{\beta}-\beta_{0}){\bf X}.
\end{align*}
It follows from Lemma~2
that, 
\begin{align*}
\frac{1}{n}\mbox{RSS}_{1} & -\frac{1}{n}\mbox{RSS}_{0}\\
= & \|\hat{\beta}-\beta_{0}\|_{\hat{\Gamma}}^{2}-\frac{2}{n}\epsilon^{T}\int(\hat{\beta}-\beta_{0}){\bf X}\\
= & \frac{1}{n^{2}}\epsilon^{T}\Big\{\int_{0}^{1}\int_{0}^{1}\hat{Q}^{+}\hat{U}(t)\hat{Q}(t,s)\hat{Q}^{+}\hat{U}(s)^{T}dtds-2\int_{0}^{1}\hat{U}(t)\hat{Q}^{+}\hat{U}(t)^{T}dt\Big\}\epsilon\\
 & -\frac{1}{n^{2}}\epsilon^{T}\widetilde{X}(1)^{T}\hat{H}^{-1}\widetilde{X}(1)\epsilon\\
= & -\frac{2}{n}\epsilon^{T}A_{n}\epsilon=o_{p}(n^{-1/2}),
\end{align*}
provided that $tr(A_{n}^{2})=o(n)$. Hence, with the fact that under $H_{0}$,
$\sigma^{2}=\mbox{RSS}_{0}/n+O_{p}(n^{-1/2})$, the likelihood ratio
test statistic $\tau_{n,\lambda}$ becomes
\begin{align*}
\tau_{n,\lambda} & =-\frac{n}{2}\log\frac{\mbox{RSS}_{1}/n}{\mbox{RSS}_{0}/n}=-\frac{n}{2\sigma^{2}}\Big(\frac{1}{n}\mbox{RSS}_{1}-\frac{1}{n}\mbox{RSS}_{0}\Big)\big(1+o_{p}(n^{-\frac{1}{2}})\big)\\
 & =z^{T}A_{n}z+o_{p}(1),
\end{align*}
where $z=\epsilon/\sigma$. 

To show that $\tau_{n,\lambda}$ has an asymptotic normal distribution
with mean $\mu_{n}=tr(A_{n})$ and variance $\sigma_{n}^{2}=2tr(A_{n}^{2})$,
we need to show that
\[
Trace(A_{n}^{4})/\sigma_{n}^{4}\rightarrow0.
\]
Let
\[
A_{I}=\frac{1}{n}\int_{0}^{1}\hat{U}(t)\hat{Q}^{+}\hat{U}(t)^{T}dt-\frac{1}{2n}\int_{0}^{1}\int_{0}^{1}\hat{Q}^{+}\hat{U}(t)\hat{Q}(t,s)\hat{Q}^{+}\hat{U}(s)^{T}dtds,
\]
and 
\[
A_{II}=\frac{1}{2}B.
\]
So $tr(A)=tr(A_{I})+tr(A_{II})$. Noting that $tr(A_{II})=m/2$, $tr(A)$
is of the same order as $tr(A_{I})$. Recall that $\hat{Q}(t,s)=\sum_{j=1}^{\infty}\hat{\kappa}_{j}\hat{\phi}_{j}(t)\hat{\phi}_{j}(s)$.
and $\hat{U}_{X_{i}}(t)=\sum_{k=1}^{\infty}\hat{\xi}_{ik}\hat{\phi}_{k}(t)$
with $n^{-1}\sum_{i=1}^{n}\hat{\xi}_{ik}^{2}=\hat{\kappa}_{k}$ and
$n^{-1}\sum_{i=1}^{n}\hat{\xi}_{ik}\hat{\xi}_{ij}=0$ for $k\neq j$.
Therefore 
\[
(A_{I})_{ij}=\frac{1}{n}\sum_{k=1}^{\infty}\frac{(2\lambda+\hat{\kappa}_{k})\hat{\xi}_{ik}\hat{\xi}_{jk}}{2(\lambda+\hat{\kappa}_{k})^{2}}.
\]
Further 
\begin{align*}
tr(A_{I}) & =\sum_{k=1}^{\infty}\frac{\hat{\kappa}_{k}(2\lambda+\hat{\kappa}_{k})}{2(\lambda+\hat{\kappa}_{k})^{2}}\asymp\sum_{k=1}^{\infty}\frac{\hat{\kappa}_{k}}{\lambda+\hat{\kappa}_{k}}.
\end{align*}
Similarly, we can show that 
\[
(A_{I}^{2})_{ij}=\frac{1}{n}\sum_{k=1}^{\infty}\frac{(2\lambda+\hat{\kappa}_{k})^{2}\hat{\kappa}_{k}\hat{\xi}_{ik}\hat{\xi}_{jk}}{4(\lambda+\hat{\kappa}_{k})^{4}},\quad\text{ \quad}(A_{I}^{4})_{ij}=\frac{1}{n}\sum_{k=1}^{\infty}\frac{(2\lambda+\hat{\kappa}_{k})^{4}(\hat{\kappa}_{k})^{3}\hat{\xi}_{ik}\hat{\xi}_{jk}}{16(\lambda+\hat{\kappa}_{k})^{8}},
\]
and
\[
tr(A_{I}^{2})\asymp\sum_{k=1}^{\infty}\frac{\hat{\kappa}_{k}^{2}}{(\lambda+\hat{\kappa}_{k})^{2}},\quad\text{ \quad}tr(A_{I}^{4})\asymp\sum_{k=1}^{\infty}\frac{\hat{\kappa}_{k}^{4}}{(\lambda+\hat{\kappa}_{k})^{4}}.
\]
Since $\frac{\hat{\kappa}_{k}^{4}}{(\lambda+\hat{\kappa}_{k})^{4}}\leq\frac{\hat{\kappa}_{k}^{2}}{(\lambda+\hat{\kappa}_{k})^{2}}$,
therefore $tr(A_{n}^{4})=O(\sigma_{n}^{2})$, and further $tr(A_{n}^{4})/\sigma_{n}^{4}\rightarrow0$.

For part (b), Under $H'_{1}$, 
\begin{align*}
\frac{1}{n}\mbox{RSS}_{0} & =\frac{1}{n}\epsilon^{T}\epsilon+\|\beta_{0}\|_{\hat{\Gamma}}^{2}+\frac{2}{n}\epsilon^{T}\int\beta_{0}{\bf X},\\
\frac{1}{n}\mbox{RSS}_{1} & =\frac{1}{n}\epsilon^{T}\epsilon+\|\hat{\beta}-\beta_{0}\|_{\hat{\Gamma}}^{2}-\frac{2}{n}\epsilon^{T}\int(\hat{\beta}-\beta_{0}){\bf X}\\
 & =\sigma^{2}-\frac{2}{n}\epsilon^{T}A\epsilon\\
 & \qquad+\lambda^{2}\int_{0}^{1}\int_{0}^{1}\hat{Q}(t,s)\hat{Q}^{+}\beta_{0}^{(m)}(t)\hat{Q}^{+}\beta_{0}^{(m)}(s)dtds\\
 & \qquad+(-1)^{m}\frac{2\lambda}{n}\epsilon^{T}\int_{0}^{1}\hat{U}(t)\hat{Q}^{+}\beta_{0}^{(m)}(t)dt-\frac{2}{n}\epsilon^{T}\int\beta_{0}{\bf X}\\
 & \qquad+(-1)^{m+1}\frac{2\lambda}{n}\epsilon^{T}\int_{0}^{1}\int_{0}^{1}\hat{Q}(t,s)\hat{Q}^{+}\beta_{0}^{(m)}(t)\hat{Q}^{+}\hat{U}(s)dtds.
\end{align*}

For $\frac{1}{n}\mbox{RSS}_{0}$,
\begin{align*}
Var(\frac{2}{n}\epsilon^{T}\int\beta_{0}{\bf X})=\frac{4}{n^{2}}Var(\sum_{i=1}^{n}\epsilon_{i}\int\beta_{0}X_{i})=\frac{4\sigma^{2}}{n^{2}}\sum_{i=1}^{n}\{\int\beta_{0}(s)X_{i}(s)ds\}^{2}=O(\frac{1}{n}||\beta_{0}||_{\hat{\Gamma}}^{2}).
\end{align*}

For $\frac{1}{n}\mbox{RSS}_{1}$, write $\beta_{0}^{(m)}(t)=\sum_{j=1}^{\infty}\hat{\eta}_{j}\hat{\phi}_{j}(t)$.
Since $\beta_{0}^{(m)}\in L_{2}$, we have $\sum_{j=1}^{\infty}\hat{\eta}_{j}^{2}<\infty$.
In the above expansion of $\frac{1}{n}\mbox{RSS}_{1}$, 
\begin{align*}
\lambda^{2}\int_{0}^{1}\int_{0}^{1}\hat{Q}(t,s) & \hat{Q}^{+}\beta_{0}^{(m)}(t)\hat{Q}^{+}\beta_{0}^{(m)}(s)dtds\\
= & \lambda^{2}\sum_{j=1}^{\infty}\frac{\hat{\kappa}_{j}\hat{\eta}_{j}^{2}}{(\lambda+\hat{\kappa}_{j})^{2}}\le\lambda^{2}\sum_{j=1}^{\infty}\hat{\eta}_{j}^{2}\sup_{x\ge0}\frac{x}{(\lambda+x)^{2}}\le\frac{\lambda}{4}\sum_{j=1}^{\infty}\hat{\eta}_{j}^{2}=O(\lambda).
\end{align*}
Further, 
\[
(-1)^{m}\frac{2\lambda}{n}\epsilon^{T}\int_{0}^{1}\hat{U}(t)\hat{Q}^{+}\beta_{0}^{(m)}(t)dt=(-1)^{m}\frac{2\lambda}{n}\sum_{i=1}^{n}\epsilon_{i}\sum_{k=1}^{\infty}\frac{\hat{\xi}_{ik}\hat{\eta}_{k}}{\lambda+\hat{\kappa}_{k}},
\]
and its variance is 
\[
\frac{4\lambda^{2}\sigma^{2}}{n}\sum_{k=1}^{\infty}\frac{\hat{\kappa}_{k}\eta_{k}^{2}}{(\lambda+\hat{\kappa}_{k})^{2}}\le\frac{4\lambda^{2}\sigma^{2}}{n}\sup_{x\ge0}\frac{x}{(\lambda+x)^{2}}\sum_{j=1}^{\infty}\hat{\eta}_{j}^{2}\le\frac{\lambda\sigma^{2}}{n}\sum_{j=1}^{\infty}\hat{\eta}_{j}^{2}=O(\lambda/n).
\]
The last term becomes 
\[
(-1)^{m+1}\frac{2\lambda}{n}\epsilon^{T}\int_{0}^{1}\int_{0}^{1}\hat{Q}(t,s)\hat{Q}^{+}\beta_{0}^{(m)}(t)\hat{Q}^{+}\hat{U}(s)dtds=(-1)^{m+1}\frac{2\lambda}{n}\sum_{i=1}^{n}\epsilon_{i}\sum_{k=1}^{\infty}\frac{\hat{\kappa}_{k}\hat{\eta}_{k}\hat{\xi}_{ik}}{(\lambda+\hat{\kappa}_{k})^{2}}.
\]
Since $\sum_{k=1}^{\infty}\frac{\hat{\kappa}_{k}\hat{\eta}_{k}\hat{\xi}_{ik}}{(\lambda+\hat{\kappa}_{k})^{2}}\leq\sum_{k=1}^{\infty}\frac{\hat{\eta}_{k}\hat{\xi}_{ik}}{\lambda+\hat{\kappa}_{k}}$,
the variance of last term is controlled by $O(\lambda/n)$. So altogether,
\[
\frac{1}{n}\mbox{RSS}_{1}-\frac{1}{n}\mbox{RSS}_{0}=-\frac{2}{n}\epsilon^{T}A\epsilon-\|\beta_{0}\|_{\hat{\Gamma}}^{2}+O\Big(\lambda\Big)+O_{p}\Big(n^{-1/2}\lambda^{1/2}\Big)+O_{p}\Big(n^{-1/2}\|\beta_{0}\|_{\hat{\Gamma}}\Big).
\]
Since $\rho_{n}^{2}=o(n^{-1/2})$ and $\lambda=o(n^{-1/2})$, therefor
$\frac{1}{n}\mbox{RSS}_{1}-\frac{1}{n}\mbox{RSS}_{0}=o_{p}(n^{-1/2})$
and 
\[
\tau_{n,\lambda}=z^{T}Az+\frac{n}{2\sigma^{2}}\|\beta_{0}\|_{\hat{\Gamma}}^{2}+O\Big(n\lambda\Big)+O_{p}\Big(n^{1/2}\lambda^{1/2}\Big)+O_{p}\Big(n^{1/2}\|\beta_{0}\|_{\hat{\Gamma}}\Big).
\]


\subsection{Proof of Theorem \ref{thm:lower}}

The proof follows \cite{ingster_93}. 
First show part (a).
Let
\begin{equation*}
\rho_{n}=n^{-2r/(1+4r)},
\end{equation*}
and suppose that $\rho_{n}'/\rho_{n}\rightarrow0$.
We show that, for any test $\phi_{n}$, 
\[
\lim\inf_{n\rightarrow\infty}\gamma_{n}(\phi_{n},\rho_{n}')\ge1.
\]
The idea of deriving the lower bound is standard. Let $\pi_{n}$ be
a probability measure on ${\cal F}_{K,\Gamma}(\rho_{n}')$. Then the
lower bound is based on the inequality 
\[
\sup_{f\in{\cal F}_{K,\Gamma}(\rho_{n}')}\mathbb{P}_{f}(\phi_{n}=0)\ge \mathbb{P}_{f,\pi_{n}}(\phi_{n}=0),
\]
where $\mathbb{P}_{f,\pi_{n}}=\int\mathbb{P}_{f}d\pi_{n}$. Write
\[
\gamma_{n,\pi_{n}}=\mathbb{P}_{0}(\phi_{n}=1)+\mathbb{P}_{f,\pi_{n}}(\phi_{n}=0).
\]
Denote by $\ell_{n,\pi_{n}}$ the likelihood ratio, 
\[
\ell_{n,\pi_{n}}=\frac{d~\mathbb{P}_{f,\pi_{n}}}{d~\mathbb{P}_{0}}=\int\frac{d~\mathbb{P}_{f}}{d~\mathbb{P}_{0}}d\pi_{n}.
\]
For any $f\in{\cal F}_{K,\Gamma}(\rho_{n})$, direct calculation yields
that 
\[
\log\frac{d~\mathbb{P}_{f}}{d~\mathbb{P}_{0}}=\frac{1}{\sigma^{2}}\sum_{i=1}^{n}Y_{i}\int X_{i}f-\frac{n}{2\sigma^{2}}\|f\|_{\hat{\Gamma}}^{2},
\]
where $\hat{\Gamma}$ is the empirical covariance function such as
\[
\hat{\Gamma}(t,s)=\frac{1}{n}\sum_{i=1}^{n}X_{i}(t)X_{i}(s).
\]
It is convenient to use the following inequalities \cite{ingster_86}:
\[
\gamma_{n,\pi_{n}}(\phi_{n},\rho_{n}')=1-\frac{1}{2}\mbox{var}(\mathbb{P}_{0},\mathbb{P}_{f,\pi_{n}})\ge1-\frac{1}{2}\delta_{n,\pi_{n}},
\]
where $\mbox{var}(\mathbb{P}_{0},\mathbb{P}_{f,\pi_{n}})$ stands
for $L_{1}$ distance between two measures, and 
\[
\delta_{n,\pi_{n}}^{2}=\mathbb{E}_{0}(\ell_{n,\pi_{n}}-1)^{2}.
\]

In the following, we select a probability measure $\pi_{n}$ for which
$\gamma_{n,\pi_{n}}$ can be effectively estimated. Recall that $K=TT^{*}$,
where $T^{*}$ is the adjoint operator to $T$ such that $\langle f,Tg\rangle=\langle T^{*}f,g\rangle$.
Define the linear operator $T\hat{\Gamma}T^{*}$ and let $\hat{s}_{1}\ge\hat{s}_{2}\ge\cdots\ge0$
be the eigenvalues of $T\hat{\Gamma}T^{*}$ and the $\hat{\varphi}_{k}$
be the corresponding eigenfunctions. Consider 
\begin{equation}
f_{\xi}=u\sum_{k=1}^{M}\xi_{k}g_{k},\label{equ:fxi}
\end{equation}
where $\xi=(\xi_{1},\ldots,\xi_{M})$ and $\xi_{k}=\pm1$ with probability
$1/2$, and $g_{k}=\hat{s}_{k}^{-1/2}T^{*}\hat{\varphi}_{k}$. In
(\ref{equ:fxi}), we choose $M=2n^{2/(4r+1)}$ and $u=n^{-1/(4r+1)}\rho_{n}'$.
Note that 
\[
\Big\langle g_{k},g_{j}\Big\rangle_{\hat{\Gamma}}=(\hat{s}_{k}\hat{s}_{j})^{-1/2}\Big\langle T^{*}\hat{\varphi}_{k},T^{*}\hat{\varphi}_{j}\Big\rangle_{\hat{\Gamma}}=(\hat{s}_{k}\hat{s}_{j})^{-1/2}\Big\langle T{\hat{\Gamma}}T^{*}\hat{\varphi}_{k},\hat{\varphi}_{j}\Big\rangle_{L_{2}}=\delta_{jk},
\]
where $\delta_{jk}=1$ for $j=k$, and $0$ for $j\neq k$. Further,
\[
\Big\langle g_{k},g_{j}\Big\rangle_{{\cal H}(K)}=(\hat{s}_{k}\hat{s}_{j})^{-1/2}\Big\langle T^{*}\hat{\varphi}_{k},T^{*}\hat{\varphi}_{j}\Big\rangle_{{\cal H}(K)}=(\hat{s}_{k}\hat{s}_{j})^{-1/2}\Big\langle\hat{\varphi}_{k},\hat{\varphi}_{j}\Big\rangle_{L_{2}}=(\hat{s}_{k}\hat{s}_{j})^{-1/2}\delta_{jk}.
\]
 It is easy to check that 
\begin{align*}
\|f_{\xi}\|_{{\cal H}(K)}^{2}=u^{2}\sum_{k=1}^{M}\hat{s}_{k}^{-1}\le u^{2}M\hat{s}_{M}^{-1}=2(\rho_{n}')^{2}s_{M}^{-1}(1+o_{p}(1)),
\end{align*}
which is bounded since $s_{M}$ has the same order with $\rho_{n}^{2}=n^{-4r/(4r+1)}$
and $\rho_{n}'/\rho_{n}=o(1)$. For any $\varphi\in L_{2}$, $T^{*}\varphi\in{\cal H}(K)$
\citep{cucker_01}. Therefore, $f_{\xi}\in{\cal H}(K)$. On the other
hand, 
\[
\|f_{\xi}\|_{\hat{\Gamma}}^{2}=Mu^{2}=2(\rho_{n}')^{2}.
\]
So, $\|f_{\xi}\|_{\Gamma}^{2}=2(\rho_{n}')^{2}(1+o(1))\ge(\rho_{n}')^{2}$
and it shows that $f_{\xi}\in{\cal F}_{K,\Gamma}(\rho_{n}')$.

For this case, the likelihood ratio is 
\begin{align*}
\ell_{n,\pi_{n}} & =\mathbb{E}_{\xi}\frac{d~\mathbb{P}_{f_{\xi}}}{d~\mathbb{P}_{0}}=\exp\Big(-\frac{nMu^{2}}{2\sigma^{2}}\Big)\mathbb{E}_{\xi}\exp(\frac{u}{\sigma^{2}}\sum_{k=1}^{M}\sum_{i=1}^{n}Y_{i}x_{ik}\xi_{k})\\
& =\exp\Big(-\frac{nMu^{2}}{2\sigma^{2}}\Big)\prod_{k=1}^{M}cosh(\frac{u}{\sigma^{2}}\sum_{k=1}^{M}Y_{i}x_{ik}).
\end{align*}
where $x_{ik}$ is denoted as $x_{ik}=\int X_{i}g_{k}$. Note that
$\sum_{i=1}^{n}x_{ik}^{2}=n\|g_{k}\|_{\hat{\Gamma}}^{2}=n$. Given
$X_{1},\ldots,X_{n}$, we have 
\begin{align*}
\mathbb{E}_{0}\Big(\ell_{n,\pi_{n}}\Big|X_{1},\ldots,X_{n}\Big) & =\mathbb{E}_{0}\mathbb{E}_{\xi}\frac{d~\mathbb{P}_{f_{\xi}}}{d~\mathbb{P}_{0}}=\mathbb{E}_{\xi}\mathbb{E}_{0}\frac{d~\mathbb{P}_{f_{\xi}}}{d~\mathbb{P}_{0}}\\
 & =\exp\Big(-\frac{nMu^{2}}{2\sigma^{2}}\Big)\mathbb{E}_{\xi}\prod_{i=1}^{n}\mathbb{E}_{0}\exp(\frac{u}{\sigma^{2}}\sum_{k=1}^{M}\xi_{k}x_{ik}\cdot Y_{i})\\
 & =\exp\Big(-\frac{nMu^{2}}{2\sigma^{2}}\Big)\mathbb{E}_{\xi}\prod_{i=1}^{n}\exp\{\frac{1}{2}\sigma^{2}(\frac{u}{\sigma^{2}}\sum_{k=1}^{M}\xi_{k}x_{ik})^{2}\}\\
 & =\exp(-\frac{nMu^{2}}{2\sigma^{2}})\exp(\frac{nMu^{2}}{2\sigma^{2}})=1.
\end{align*}
Noting that
\begin{align*}
\ell_{n,\pi_{n}}^{2} & =\exp\Big(-\frac{nMu^{2}}{\sigma^{2}}\Big)\prod_{k=1}^{M}cosh(\frac{u}{\sigma^{2}}\sum_{k=1}^{M}Y_{i}x_{ik})^{2}\\
 & =\exp\Big(-\frac{nMu^{2}}{\sigma^{2}}\Big)\prod_{k=1}^{M}(\frac{1}{4}e^{\frac{2u}{\sigma^{2}}\sum_{k=1}^{M}Y_{i}x_{ik}}+\frac{1}{4}e^{-\frac{2u}{\sigma^{2}}\sum_{k=1}^{M}Y_{i}x_{ik}}+\frac{1}{2})\\
 & =\exp\Big(-\frac{nMu^{2}}{\sigma^{2}}\Big)\mathbb{E}_{\varsigma}\exp(\frac{2u}{\sigma^{2}}\sum_{k=1}^{M}\sum_{i=1}^{n}Y_{i}x_{ik}\varsigma_{k}),
\end{align*}
where random variable $\varsigma$ takes value $-1,0,1$ with probability
$1/4,\,1/2,\,1/4$. Therefore we can calculate $\mathbb{E}_{0}\Big(\ell_{n,\pi_{n}}^{2}\Big|X_{1},\ldots,X_{n}\Big)$
as 
\begin{align*}
\mathbb{E}_{0}\Big(\ell_{n,\pi_{n}}^{2}\Big|X_{1},\ldots,X_{n}\Big) & =\mathbb{E}_{0}\exp\Big(-\frac{nMu^{2}}{\sigma^{2}}\Big)\mathbb{E}_{\varsigma}\exp(\frac{2u}{\sigma^{2}}\sum_{k=1}^{M}\sum_{i=1}^{n}Y_{i}x_{ik}\varsigma_{k})\\
 & =\exp\Big(-\frac{nMu^{2}}{\sigma^{2}}\Big)\mathbb{E}_{\varsigma}\prod_{i=1}^{n}\mathbb{E}_{0}\exp(\frac{2u}{\sigma^{2}}\sum_{k=1}^{M}x_{ik}\varsigma_{k}\cdot Y_{i})\\
 & =\exp\Big(-\frac{nMu^{2}}{\sigma^{2}}\Big)\mathbb{E}_{\varsigma}\exp(\frac{2u^{2}n}{\sigma^{2}}\sum_{k=1}^{M}\varsigma_{k}^{2})\\
 & =\exp\Big(-\frac{nMu^{2}}{\sigma^{2}}\Big)\prod_{k=1}^{M}\{\frac{1}{2}\exp(\frac{2u^{2}n}{\sigma^{2}})+\frac{1}{2}\}\\
 & =\{\cosh(\frac{u^{2}n}{\sigma^{2}})\}^{M},
\end{align*}
and
\begin{align*}
\mathbb{E}_{0}(\ell_{n,\pi_{n}}-1)^{2} & =\mathbb{E}_{0}\Big(\ell_{n,\pi_{n}}^{2}\Big|X_{1},\ldots,X_{n}\Big)-2\mathbb{E}_{0}\Big(\ell_{n,\pi_{n}}\Big|X_{1},\ldots,X_{n}\Big)+1\\
 & =\{\cosh(\frac{u^{2}n}{\sigma^{2}})\}^{M}-1.
\end{align*}
Using the inequality $\log\cosh x\le Bx^{2}$ for a certain $B$,
\begin{align*}
\{\cosh(\frac{u^{2}n}{\sigma^{2}})\}^{M}-1 & \le\exp\Big(\frac{BMu^{4}n^{2}}{\sigma^{4}}\Big)-1.
\end{align*}
Hence 
\[
\mathbb{E}_{0}(\ell_{n,\pi_{n}}-1)^{2}\le\exp\Big(\frac{Bn^{2}Mu^{4}}{\sigma^{4}}\Big)-1.
\]
Our choices of $M$ and $u$ guarantees that $n^{2}Mu^{4}\rightarrow0$,
so $\lim\inf_{n\rightarrow}\gamma_{n}(\phi_{n},\rho_{n}')=1$. This
completes the proof of part (a).

Next, we prove part (b). The proof is similar. In particular, in (\ref{equ:fxi})
we choose $M=\log n/(2r)$ and $u=2\rho_{n}'\sqrt{2r/\log n}$, where
$\rho_{n}'/\rho_{n}\rightarrow0$ with $\rho_{n}=n^{-1/2}(\log n/(2r))^{1/4}$.
It is easy to see that $n^{2}Mu^{4}\rightarrow0$, so that $\lim\inf_{n\rightarrow}\gamma_{n}(\phi_{n},\rho_{n}')=1$.
This completes the proof of part (b).

\subsection{Proof of Theorem \ref{thm:ttheory}}

Recall that $H'_{1}:{\cal F}'_{K,\Gamma}(\rho_{n})=\big\{\beta\in\mathcal{H}(K):~\|\beta\|_{\Gamma}=\rho_{n}\big\}$,
we only need to show that 
\[
\lim_{c_{n}\rightarrow\infty}\inf_{\beta_{0}\in{\cal F}'_{K,\Gamma}(c_{n}\rho_{n})}\mathbb{P}_{\beta_{0}}\Big(\frac{\tau_{n,\lambda}-\mu_{n}}{\sigma_{n}}>z_{\alpha}\Big)=1.
\]
The power function under $H_{1}'$ can be written as
\begin{align*}
 & \mathbb{P}_{\beta_{0}}\Big(\frac{\tau_{n,\lambda}-\mu_{n}}{\sigma_{n}}\ge z_{\alpha}\Big)\\
= & \mathbb{P}_{\beta_{0}}\Big\{\frac{z^{T}Az-\mu_{n}}{\sigma_{n}}+\frac{\frac{n}{2\sigma^{2}}\|\beta_{0}\|_{\hat{\Gamma}}^{2}+O\Big(n\lambda\Big)+O_{p}\Big(n^{1/2}\lambda^{1/2}\Big)+O_{p}\Big(n^{1/2}\|\beta_{0}\|_{\hat{\Gamma}}\Big)}{\sigma_{n}}\ge z_{\alpha}\Big\}.
\end{align*}
 Recall that $\sigma_{n}^{2}=tr(A^{2})=O(tr(A))$ as shown in the
proof as Theorem 2, 
and by Lemma~3,
we have 
\[
\mu_{n}=O_{p}(\sum_{k=1}^{\infty}\frac{s_{k}}{\lambda+s_{k}})\quad\quad and\quad\quad\sigma_{n}^{2}=O_{p}(\sum_{k=1}^{\infty}\frac{s_{k}}{\lambda+s_{k}}).
\]
Therefore $\mu_{n}$ and $\sigma_{n}^{2}$ are of order $O_{p}(\lambda^{-1/2r}$)
when $s_{k}\asymp k^{-2r}$, or of order $O_{p}\{(2r)^{-1}\log\lambda^{-1}\}$
when $s_{k}\asymp e^{-2rk}$. Recall that when $\kappa_{k}\asymp k^{-2r}$,
the optimal $\lambda$ is of order $n^{-4r/(4r+1)}$; when $\kappa_{k}\asymp exp^{-2rk}$,
$\log\lambda^{-1}=O(\log n)$. So, when $\kappa_{k}\asymp k^{-2r}$,
\[
\lim_{c_{n}\rightarrow\infty}\inf_{\beta\in{\cal F}_{K,\Gamma}:\|\beta\|_{\Gamma}\ge c_{n}n^{-2r/(4r+1)}}\mathbb{P}_{\beta}\Big(\frac{\tau_{n,\lambda}-\mu_{n}}{\sigma_{n}}\ge z_{\alpha}\Big)=1,
\]
and when $\kappa_{k}\asymp e^{-2rk}$, 
\[
\lim_{c_{n}\rightarrow\infty}\inf_{\beta\in{\cal F}_{K,\Gamma}:\|\beta\|_{\Gamma}\ge c_{n}\{\log n/(2rn^{2})\}^{1/4}}\mathbb{P}_{\beta}\Big(\frac{\tau_{n,\lambda}-\mu_{n}}{\sigma_{n}}\ge z_{\alpha}\Big)=1.
\]
This finishes the proof of the theorem.


\subsection{Proof of Theorem \ref{thm:adaptive}}

First noting that $s_{k}$ and $\kappa_{k}$ have the same decay rate,
so we can replace $s_{k}$ in condition $s_{k}\asymp k^{-2r}$ 
by $\kappa_{k}$. 

Given a symmetric bivariate function $M$, let $|||M|||=(\int\int M^{2})^{1/2}$.
Define $\delta_{k}=\min_{1\le j\le k}(\kappa_{j}-\kappa_{j+1})$ which
is of order $k^{-2r-1}$. $\tilde{\Delta}=|||\tilde{Q}-Q|||$, $\tilde{\Delta}_{j}=\|\int(\tilde{Q}-Q)\phi_{j}\|$,
and 
\[
\tilde{\Delta}_{jj}=\int_{0}^{1}\int_{0}^{1}(\tilde{Q}(t,s)-Q(t,s))\phi_{j}(t)\phi_{j}(s)dtds.
\]
It follows from Equation (5.7) of \cite{hall_07} that 
\[
\Big|\tilde{\kappa}_{j}-\kappa_{j}-\tilde{\Delta}_{jj}\Big|\le\delta_{j}^{-1}\tilde{\Delta}(\tilde{\Delta}+\tilde{\Delta}_{j}),
\]
and we also have $\mathbb{E}\tilde{\Delta}_{jj}^{2}\le C_{1}n^{-1}\kappa_{j}^{2}$
and $\mathbb{E}(\tilde{\Delta}^{2}+\tilde{\Delta}_{j}^{2})\le C_{2}n^{-1}$
where $C_{1}$ and $C_{2}$ do not depend on $j$. Observe that 
\[
\sum_{j=1}^{\varrho}|\tilde{\kappa}_{j}-\kappa_{j}|\le\sum_{j=1}^{\varrho}|\tilde{\Delta}_{j}|+\tilde{\Delta}\sum_{j=1}^{\varrho}\delta_{j}^{-1}(\tilde{\Delta}+\tilde{\Delta}_{j}).
\]
Further, $\sum_{j=1}^{\varrho}|\tilde{\Delta}_{jj}|$ is of order
$O_{p}(n^{-1/2})$ since 
\[
\mathbb{E}\sum_{j=1}^{\varrho}|\tilde{\Delta}_{jj}|\le\sum_{j=1}^{\varrho}\sqrt{\mathbb{E}\tilde{\Delta}_{jj}^{2}}\le C_{1}n^{-1/2}\sum_{j=1}^{\varrho}\kappa_{j}=O(n^{-1/2}),
\]
and $\tilde{\Delta}\sum_{j=1}^{\varrho}\delta_{j}^{-1}(\tilde{\Delta}+\tilde{\Delta}_{j})$
is or order $O_{p}(n^{-1}\varrho^{2r+2})$ since 
\[
\mathbb{E}\sum_{j=1}^{\varrho}|\delta_{j}^{-1}(\tilde{\Delta}+\tilde{\Delta}_{j})|\le\sum_{j=1}^{\varrho}\delta_{j}^{-1}\sqrt{2\mathbb{E}(\tilde{\Delta}^{2}+\tilde{\Delta}_{j}^{2})}\le\sqrt{2C_{2}n^{-1}}\sum_{j=1}^{\varrho}\delta_{j}^{-1}=O(n^{-1/2}\varrho^{2r+2}).
\]
Hence, 
\[
\sum_{j=1}^{\varrho}\Big|\tilde{\kappa}_{j}-\kappa_{j}\Big|=O_{p}\Big(n^{-1/2}+n^{-1}\varrho^{2r+2}\Big).
\]
On the other hand, since $E(\tilde{Q}-Q)^{2}=O(n^{-1})$ uniformly
on $[0,1]^{2}$, 
\begin{align*}
\Big|\sum_{j=\varrho+1}^{\infty}(\tilde{\kappa}_{j}-\kappa_{j})\Big| & =\Big|\int\int(\tilde{Q}-Q)(s,t)dsdt-\sum_{j=1}^{\varrho}(\tilde{\kappa}_{j}-\kappa_{j})\Big|\\
 & \le[\int\int(\tilde{Q}-Q)^{2}]^{1/2}+\Big|\sum_{j=1}^{\varrho}\Big(\tilde{\kappa}_{j}-\kappa_{j}\Big)\Big|\\
 & =O_{p}(n^{-1/2}+n^{-1}\varrho^{2r+2}).
\end{align*}
If we choose $\rho\asymp n^{1/(4r+1)}$, we have 
\[
\Big|\sum_{j=1}^{\varrho}(\tilde{\kappa}_{j}-\kappa_{j})\Big|=O_{p}(n^{(-2r+1)/(4r+1)}),~~~~\Big|\sum_{j=\varrho+1}^{\infty}(\tilde{\kappa}_{j}-\kappa_{j})\Big|=O_{p}(n^{(-2r+1)/(4r+1)}).
\]

Define the event ${\cal E}_{\varrho}$ by 
\[
{\cal E}_{\varrho}={\cal E}_{\varrho}(n)=\{\frac{1}{2}\kappa_{\varrho}\ge\tilde{\Delta}\}.
\]
Since $\sup_{k\ge1}|\tilde{\kappa}_{k}-\kappa_{k}|\le\tilde{\Delta}$
\cite{bhatia_83}, if ${\cal E}_{\varrho}$ holds, we have $\tilde{\kappa}_{k}\ge\frac{1}{2}\kappa_{k}$
for $1\le k\le\varrho$. Here, we choose $\varrho\asymp n^{1/(4r+1)}$,
which implies that $n^{1/2}\kappa_{\varrho}\rightarrow\infty$ as
$n\rightarrow\infty$. Since $\tilde{\Delta}=O_{p}(n^{-1/2})$, we
have $\mathbb{P}({\cal E}_{\varrho})\rightarrow1$. Therefore, since
the result we wish to prove only relates to probabilities of differences
(not to moments of differences), it suffices to work with bounds that
are established under the assumption that ${\cal E}_{\varrho}$ holds.
The optimal choice $\tilde{\lambda}$ is the root of 
\[
\frac{1}{n}\sum_{k=1}^{\infty}\frac{\tilde{\kappa}_{k}}{(\sqrt{\lambda}+\tilde{\kappa}_{k})^{2}}=2\sqrt{\lambda}.
\]
In the following, we derive the asymptotic order of $\sum_{k=1}^{\infty}\frac{\tilde{\kappa}_{k}}{(\sqrt{\lambda}+\tilde{\kappa}_{k})^{2}}$.
Note that 
\begin{align}
\sum_{k=1}^{\infty}\frac{\tilde{\kappa}_{k}}{(\sqrt{\lambda}+\tilde{\kappa}_{k})^{2}} & =\sum_{k=1}^{\varrho}\frac{\tilde{\kappa}_{k}}{(\sqrt{\lambda}+\tilde{\kappa}_{k})^{2}}+\sum_{k=\varrho+1}^{\infty}\frac{\tilde{\kappa}_{k}}{(\sqrt{\lambda}+\tilde{\kappa}_{k})^{2}}\nonumber \\
 & \le\sum_{k=1}^{\varrho}\tilde{\kappa}_{k}^{-1}+\lambda^{-1}\sum_{k=\varrho+1}^{\infty}\tilde{\kappa}_{k}\nonumber \\
 & \le2\sum_{k=1}^{\varrho}\kappa_{k}+\lambda^{-1}\sum_{k=\varrho+1}^{\infty}\kappa_{k}+\lambda^{-1}\Big|\sum_{j=\varrho+1}^{\infty}(\tilde{\kappa}_{j}-\kappa_{j})\Big|\nonumber \\
 & =O_{p}\Big(n^{(2r+1)/(4r+1)}+\lambda^{-1}n^{(-2r+1)/(4r+1)}\Big).\label{eq:high}
\end{align}
We also need the lower bound for $\sum_{k=1}^{\infty}\frac{\tilde{\kappa}_{k}}{(\sqrt{\lambda}+\tilde{\kappa}_{k})^{2}}$.
This follows from 
\begin{align}
\sum_{k=1}^{\infty}\frac{\tilde{\kappa}_{k}}{(\sqrt{\lambda}+\tilde{\kappa}_{k})^{2}} & \ge\sum_{k=\varrho+1}^{\infty}\frac{\tilde{\kappa}_{k}}{(\sqrt{\lambda}+\tilde{\kappa}_{k})^{2}}\nonumber \\
 & \ge\frac{1}{(\sqrt{\lambda}+\tilde{\kappa}_{\varrho})^{2}}\sum_{k=\varrho+1}^{\infty}\tilde{\kappa}_{k}\nonumber \\
 & \ge\frac{1}{2(\lambda+O_{p}(n^{-4r/(4r+1)}))}O_{p}(n^{(-2r+1)/(4r+1)}).\label{eq:low}
\end{align}
Combining (\ref{eq:high}) and (\ref{eq:low}), we obtain that $\tilde{\lambda}$
is of order $O_{p}(n^{4r/(4r+1)})$.


\section{Proof of Propositions and Lemmas \label{sec:Proof-of-Propositions}}

\noindent \textit{Proof of Proposition \ref{prop:kt}}.
Let ${\cal D}_{k}=\mbox{span}\{f_{1},\ldots,f_{k}:K^{1/2}f_{j}=T^{*}\varphi_{j},j=1,\ldots,k\}$.
It follows from the minimax principal that 
\begin{align*}
\tilde{s}_{k} & \le\sup_{f\in{\cal D}_{k}^{\perp}}\frac{\Big\langle K^{1/2}\Gamma K^{1/2}f,~f\Big\rangle}{\langle f,f\rangle}=\sup_{f\in{\cal D}_{k}^{\perp}}\frac{\Big\langle\Gamma K^{1/2}f,~K^{1/2}f\Big\rangle}{\langle f,f\rangle}\\
 & =\sup_{g\in\tilde{{\cal D}}_{k}^{\perp}}\frac{\Big\langle\Gamma T^{*}g,~T^{*}g\Big\rangle}{\langle g,g\rangle}\frac{\langle g,g\rangle}{\langle T^{*}g,T^{*}g\rangle}\le cs_{k},
\end{align*}
where $\tilde{s}_{k}$ is the $k$th eigenvalue of $K^{1/2}\Gamma K^{1/2}$,
$\tilde{{\cal D}}_{k}=\mbox{span}\{\varphi_{1},\ldots,\varphi_{k}\}$
and the constant $c>0$ does not depend on $k$. Using a similar argument,
we may show that $s_{k}\le c\tilde{s}_{k}$. Therefore, the eigenvalues
of $T\Gamma T^{*}$ and $K^{1/2}\Gamma K^{1/2}$ have the same decay
rate. 
\qed 

\begin{lemma}
\label{lem:psi} The following statements are true: 

(a). The $\beta\in W_{2}^{m}$ minimizes $L(\beta)$, if and only
if, $L_{1}(\beta,\beta_{1})=0$ for all $\beta_{1}\in W_{2}^{m}$. 

(b). If $\beta\in W_{2}^{m}$ minimizes $L(\beta)$, then for all
$\beta_{1}\in W_{2}^{m}$, 
\begin{equation}
L_{1}(\beta,\beta_{1})=\int_{0}^{1}L_{2}(\beta)(t)~\beta_{1}^{(m)}(t)dt,\label{equ:psi2-1}
\end{equation}
where 
\begin{equation}
L_{2}(\beta)=(\lambda I+\hat{Q})\beta^{(m)}-\frac{(-1)^{m}}{n}\hat{U}^{T}{\bf Y}.\label{equ:psi2-2}
\end{equation}
\end{lemma}

\begin{proof}
First show part (a). If $\hat{\beta}\in W_{2}^{m}$ minimizes $L(\beta)$,
then $L(\hat{\beta}+\delta\beta_{1})-L(\hat{\beta})\ge0$ for all
$\beta_{1}\in W_{2}^{m}$ and any $\delta\in\mathbb{R}$. Then $L_{1}(\hat{\beta},\beta_{1})=0$
follows since $\delta$ can be either negative or positive. On the
other hand, if $L_{1}(\hat{\beta},\beta_{1})=0$, we have $L(\hat{\beta}+\delta\beta_{1})-L(\hat{\beta})\ge0$
by (\ref{equ:temp1}). Thus, $\hat{\beta}$ minimizes $L(\beta)$.
Therefore, part (a) follows.

Let $\beta_{1}(t)=t^{(k-1)}$, $k=1,\ldots,m$ in (\ref{equ:psi1}).
If $\hat{\beta}$ minimizes $L(\beta)$, then
\begin{equation}
\frac{1}{n}\sum_{i=1}^{n}\big\{ Y_{i}-\int_{0}^{1}X_{i}(s)\hat{\beta}(s)ds\big\}\big\{\int_{0}^{1}X_{i}(s)s^{(k-1)}ds\big\}=0.\label{equ:temp2}
\end{equation}
Let $X_{i}^{(-k)}(t)=T_{0}^{k}X_{i}(t)=\int_{0}^{1}\frac{(t-s)_{+}^{(k-1)}}{(k-1)!}X_{i}(s)ds$.
When $k=1$, $\int_{0}^{1}X_{i}(s)s^{(k-1)}ds=X_{i}^{(-1)}(1)$
and further (\ref{equ:temp2}) becomes 
\[
\frac{1}{n}\sum_{i=1}^{n}X_{i}^{(-1)}(1)\big\{ Y_{i}-\int_{0}^{1}X_{i}(s)\hat{\beta}(s)ds\big\}=0.
\]
When $k=2$, we have 
\begin{align*}
\frac{1}{n}\sum_{i=1}^{n} & \big\{ Y_{i}-\int_{0}^{1}X_{i}(s)\hat{\beta}(s)ds\big\}\big\{\int_{0}^{1}X_{i}(s)sds\big\}\\
 & =-\frac{1}{n}\sum_{i=1}^{n}\big\{ Y_{i}-\int_{0}^{1}X_{i}(s)\hat{\beta}(s)ds\big\}\big\{\int_{0}^{1}X_{i}(s)(1-s)ds\big\}\\
 & =-\frac{1}{n}\sum_{i=1}^{n}\big\{ Y_{i}-\int_{0}^{1}X_{i}(s)\hat{\beta}(s)ds\big\}\big\{ X_{i}^{(-2)}(1)\big\}.
\end{align*}
Hence, 
\[
\frac{1}{n}\sum_{i=1}^{n}X_{i}^{(-2)}(1)\big\{ Y_{i}-\int_{0}^{1}X_{i}(s)\hat{\beta}(s)ds\big\}=0.
\]
Following the same procedure, it can be shown that 
\begin{equation}
\frac{1}{n}\sum_{i=1}^{n}X_{i}^{(-k)}(1)\big\{ Y_{i}-\int_{0}^{1}X_{i}(s)\hat{\beta}(s)ds\big\}=0,~k=1,\ldots,m.\label{equ:opt2}
\end{equation}
Considering that 
\[
\beta_{1}(s)=\sum_{k=0}^{m-1}(-1)^{k}\frac{\beta_{1}^{(k)}(1)}{k!}(1-s)^{k}+(-1)^{m}\int_{0}^{1}\frac{\beta_{1}^{(m)}(t)}{(m-1)!}(t-s)_{+}^{m-1}dt.
\]
Therefore

\begin{align}
\int_{0}^{1} & X_{i}(s)\beta_{1}(s)ds\nonumber \\
 & =\sum_{k=0}^{m-1}(-1)^{k}\beta_{1}^{(k)}(1)\int_{0}^{1}X_{i}(s)\frac{(1-s)^{k}}{k!}ds\nonumber\\
 & \quad \quad +(-1)^{m}\int_{0}^{1}\int_{0}^{1}X_{i}(s)\frac{\beta_{1}^{(m)}(t)}{(m-1)!}(t-s)_{+}^{m-1}dtds\nonumber \\
 & =\sum_{k=1}^{m}(-1)^{k-1}\beta_{1}^{(k-1)}(1)~X_{i}^{(-k)}(1)+(-1)^{m}\int_{0}^{1}\beta_{1}^{(m)}(t)X_{i}^{(-m)}(t)dt. \label{equ:Xbtaylor}
\end{align}
If (\ref{equ:opt2}) holds, direct calculation yields 
\begin{align*}
\frac{1}{n}\sum_{i=1}^{n} & \big\{ Y_{i}-\int_{0}^{1}X_{i}(s)\hat{\beta}(s)ds\big\}\big\{\int_{0}^{1}X_{i}(s)\beta_{1}(s)ds\big\}\\
 & =\frac{(-1)^{m}}{n}\sum_{i=1}^{n}\big\{ Y_{i}-\int_{0}^{1}X_{i}(s)\hat{\beta}(s)ds\big\}\big\{\int_{0}^{1}X_{i}^{(-m)}(t)\beta_{1}^{(m)}(t)dt\big\}.
\end{align*}
Recall the definition of $L_{2}(\beta)$, we have 
\begin{align*}
L_{2}(\hat{\beta}) & =\lambda~\hat{\beta}^{(m)}(t)+\frac{(-1)^{m}}{n}\sum_{i=1}^{n}X_{i}^{(-m)}(t)\big\{\int_{0}^{1}X_{i}(s)\hat{\beta}(s)ds-Y_{i}\big\}\\
 & =\lambda~\hat{\beta}^{(m)}(t)+\frac{(-1)^{m}}{n}T_{0}^{(m)}{\bf X}(t)^{T}\Big\{\int_{0}^{1}{\bf X}(s)\hat{\beta}(s)ds-{\bf Y}\Big\}.
\end{align*}
Similar to (\ref{equ:Xbtaylor}), 
\[
\int_{0}^{1}X_{i}(s)\hat{\beta}(s)ds=\widehat{\Upsilon}(1)^{T}\widetilde{X}_{i}(1)+(-1)^{m}\int_{0}^{1}X_{i}^{(-m)}(s)\hat{\beta}^{(m)}(s)ds,~~~j=1,\ldots,m.
\]
which gives 
\[
\int_{0}^{1}{\bf X}(s)\hat{\beta}(s)ds=\widetilde{X}(1)^{T}\widehat{\Upsilon}(1)+(-1)^{m}\int_{0}^{1}T_{0}^{m}{\bf X}(s)\hat{\beta}^{(m)}(s)ds.
\]
This, combining (\ref{equ:opt2}), gives 
\begin{equation}
\hat{H}\widehat{\Upsilon}(1)+\frac{(-1)^{m}}{n}\widetilde{X}(1)\int_{0}^{1}T_{0}^{m}{\bf X}(s)\hat{\beta}^{(m)}(s)ds=\frac{1}{n}\widetilde{X}(1){\bf Y}.\label{equ:beta11}
\end{equation}
So for $\beta\in W_{2}^{m}$ minimizes $L(\beta)$, 
\begin{align*}
L_{2}(\beta)=\lambda~\beta^{(m)}+\hat{Q}\beta^{(m)}-\frac{(-1)^{m}}{n}\hat{U}^{T}{\bf Y}.
\end{align*}
So, part (b) follows.
\end{proof}

\begin{lemma}
\label{lem:expansion} Let $\epsilon=(\epsilon_{1},\ldots,\epsilon_{n})^{T}$.
The following statements hold: 

(a)
\begin{align}
\int_{0}^{1} & {\bf X}(t)\Big\{\hat{\beta}(t)-\beta_{0}(t)\Big\} dt\label{equ:expansion0}\\
 & =(-1)^{m+1}\lambda\int_{0}^{1}\hat{U}(t)\hat{Q}^{+}\beta_{0}^{(m)}(t)dt+\frac{1}{n}\Big\{\int_{0}^{1}\hat{U}(t)\hat{Q}^{+}\hat{U}(t)^{T}dt+\widetilde{X}(1)^{T}\hat{H}^{-1}\widetilde{X}(1)\Big\}\epsilon;\nonumber 
\end{align}

(b)
\begin{align}
\Big\|\hat{\beta}-\beta_{0}\Big\|_{\hat{\Gamma}}^{2} & =\lambda^{2}\int_{0}^{1}\int_{0}^{1}\hat{Q}(t,s)\hat{Q}^{+}\beta_{0}^{(m)}(t)\hat{Q}^{+}\beta_{0}^{(m)}(s)dtds\label{equ:expansion}\\
 & +\frac{1}{n^{2}}\epsilon^{T}\Big\{\int_{0}^{1}\int_{0}^{1}\hat{Q}^{+}\hat{U}(s)\hat{Q}(t,s)\hat{Q}^{+}\hat{U}(t)^{T}dsdt+\widetilde{X}(1)^{T}\hat{H}^{-1}\widetilde{X}(1)\Big\}\epsilon\nonumber \\
 & +(-1)^{m+1}\frac{2\lambda}{n}\epsilon^{T}\int_{0}^{1}\int_{0}^{1}\hat{Q}(t,s)\hat{Q}^{+}\beta_{0}^{(m)}(t)\hat{Q}^{+}\hat{U}(s)dtds.\nonumber 
\end{align}

\end{lemma}

\begin{proof}
Denote 
\[
\Upsilon_{0}(1)=\Big[\beta_{0}(1),-\beta_{0}'(1),\ldots,(-1)^{m-1}\beta_{0}^{(m-1)}(1)\Big]^{T}.
\]
Direct calculation yields 
\[
\frac{1}{n}\widetilde{X}(1){\bf Y}=\hat{H}\Upsilon_{0}(1)+(-1)^{m}\frac{1}{n}\widetilde{X}(1)\int_{0}^{1}T_{0}^{m}{\bf X}(s)\beta_{0}^{(m)}(s)ds+\frac{1}{n}\widetilde{X}(1)\epsilon.
\]
Combining this with (\ref{equ:beta11}) gives 
\[
\widehat{\Upsilon}(1)-\Upsilon_{0}(1)=(-1)^{m+1}\frac{1}{n}\hat{H}^{-1}\widetilde{X}(1)\int_{0}^{1}T_{0}^{m}X(s)\Big\{\hat{\beta}^{(m)}(s)-\beta_{0}^{(m)}(s)\Big\} ds+\frac{1}{n}\hat{H}^{-1}\widetilde{X}(1)\epsilon.
\]
Therefore, 
\begin{align}
\int_{0}^{1} & {\bf X}(s)\Big\{\hat{\beta}(s)-\beta_{0}(s)\Big\} ds\nonumber \\
 & =\widetilde{X}(1)^{T}\Big\{\widehat{\Upsilon}(1)-\Upsilon_{0}(1)\Big\}+(-1)^{m}\int_{0}^{1}T_{0}^{m}{\bf X}(s)\Big\{\hat{\beta}^{(m)}(s)-\beta_{0}^{(m)}(s)\Big\} ds\nonumber \\
 & =(-1)^{m}\int_{0}^{1}\hat{U}(s)~\Big\{\hat{\beta}^{(m)}(s)-\beta_{0}^{(m)}(s)\Big\}ds+\frac{1}{n}\widetilde{X}(1)^{T}\hat{H}^{-1}\widetilde{X}(1)\epsilon.\label{equ:pd}
\end{align}
Recall that $\hat{Q}^{+}=(\lambda I+\hat{Q})^{-1}$. It follows from
Theorem 1 
that 
\begin{align*}
\hat{\beta}^{(m)}-\beta_{0}^{(m)} & =(-1)^{m}n^{-1}{\bf Y}^{T}\hat{Q}^{+}\hat{U}-\beta_{0}^{(m)}\nonumber \\
 & =\frac{(-1)^{m}}{n}\hat{Q}^{+}\hat{U}^{T}\{\int_{0}^{1}{\bf X}(s)\beta_{0}(s)ds\}-\beta_{0}^{(m)}+(-1)^{m}\frac{1}{n}\epsilon^{T}\hat{Q}^{+}\hat{U}\nonumber \\
 & =\frac{(-1)^{m}}{n}\hat{Q}^{+}\hat{U}^{T}\widetilde{X}(1)^{T}\Upsilon_{0}(1)+\hat{Q}^{+}\hat{Q}\beta_{0}^{(m)}-\beta_{0}^{(m)}+(-1)^{m}\frac{1}{n}\epsilon^{T}\hat{Q}^{+}\hat{U}\nonumber \\
 & =-\lambda\hat{Q}^{+}\beta_{0}^{(m)}+(-1)^{m}\frac{1}{n}\epsilon^{T}\hat{Q}^{+}\hat{U}.
\end{align*}
The last equation follows from the fact that $\widetilde{X}(1)\hat{U}(s)=0$.
Then, this, combing with (\ref{equ:pd}), leads to part (a).
Furthermore, 
part (b) follows that 
\[
\Big\|\hat{\beta}-\beta_{0}\Big\|_{\hat{\Gamma}}^{2}=\frac{1}{n}\Big[\int_{0}^{1}{\bf X}(t)^{T}\Big\{\hat{\beta}(t)-\beta_{0}(t)\Big\} dt\Big]~\Big[\int_{0}^{1}{\bf X}(s)\Big\{\hat{\beta}(s)-\beta_{0}(s)\Big\} ds\Big].
\]
This completes the proof of the lemma.\end{proof}

\begin{lemma}
\label{lem: TrA} If $\lambda^{-1}=O(n)$, then $tr(A)$ is of the
same order of $\sum_{k=1}^{\infty}\frac{s_{k}}{\lambda+s_{k}}$. 
\end{lemma}

\begin{proof}
In Theorem 2,
we have shown that 
\[
tr(A)=O(\sum_{k=1}^{\infty}\frac{\hat{\kappa}_{k}}{\lambda+\hat{\kappa}_{k}}).
\]
Define that $\tilde{Q}=T_{0}^{m}\hat{\Gamma}T_{1}^{m}$. Noting that
the eigenvalues of $\hat{Q}=T_{0}^{m}\big(\hat{\Gamma}-\hat{\Gamma}_{0}\big)T_{1}^{m}$
and $\tilde{Q}$ have the same decay rate. If we write $\tilde{Q}(t,s)=\sum_{j=1}^{\infty}\tilde{\kappa}_{j}\tilde{\phi}_{j}(t)\tilde{\phi}_{j}(s)$,
then $tr(A)$ is of the same order as $\sum_{k=1}^{\infty}\frac{\tilde{\kappa}_{k}}{\lambda+\tilde{\kappa}_{k}}$.
On the other hand, recall that linear operator $Q=T_{0}^{m}\Gamma T_{1}^{m}$.
Following spectral theorem, we have $Q(t,s)=\sum_{j=1}^{\infty}\kappa_{j}\phi_{j}(t)\phi_{j}(s)$.
$\{\kappa_{k}\}$ and $\{s_{k}\}$ have the same decay rate. So we
only need to show that $\sum_{k=1}^{\infty}\frac{\tilde{\kappa}_{k}}{\lambda+\tilde{\kappa}_{k}}=O_{p}(\sum_{k=1}^{\infty}\frac{\kappa_{k}}{\lambda+\kappa_{k}})$.

Let $Q^{+}=(Q+\lambda I)^{-1}$ and $\tilde{Q}^{+}=(\tilde{Q}+\lambda I)^{-1}$.
It is easy to see that $Q^{+}(t,s)=\sum_{j=1}^{\infty}\frac{1}{\lambda+\kappa_{j}}\phi_{j}(t)\phi_{j}(s)$
and $\tilde{Q}^{+}(t,s)=\sum_{j=1}^{\infty}\frac{1}{\lambda+\tilde{\kappa}_{j}}\tilde{\phi}_{j}(t)\tilde{\phi}_{j}(s)$.
Then 
\begin{align*}
\sum_{k=1}^{\infty}\frac{\tilde{\kappa}_{k}}{\lambda+\tilde{\kappa}_{k}} & =\int\int\tilde{Q}(s,t)\tilde{Q}^{+}(s,t)dsdt\\
 & =\int\int Q(s,t)Q^{+}(s,t)dsdt\\
 & \quad\quad+\int\int Q^{+}(s,t)(\tilde{Q}-Q)(s,t)dsdt+\int\int(\tilde{Q}^{+}-Q^{+})(s,t)Q(s,t)dsdt\\
 & \quad\quad+\int\int(\tilde{Q}-Q)(s,t)(\tilde{Q}^{+}-Q^{+})(s,t)dsdt.
\end{align*}
We are going to show that all four terms above in the last equation
are either of the same order of or of $\sum_{k=1}^{\infty}\frac{\kappa_{k}}{\lambda+\kappa_{k}}$
or smaller than that.


For the first term, it is easy to see that 
\begin{align*}
\int\int Q(s,t)Q^{+}(s,t)dsdt & =\sum_{k=1}^{\infty}\frac{\kappa_{k}}{\lambda+\kappa_{k}}.
\end{align*}

For the second term, let $\Delta(s,t)=(\tilde{Q}-Q)(s,t)$ and $\hat{\Delta}_{jk}=|\int\int\Delta(s,t)\phi_{j}(s)\phi_{k}(t)dsdt|$.
It follows Section 5.3 of \cite{hall_07} that
\[
\hat{\Delta}_{jj}=|\int\int\Delta(s,t)\phi_{j}(s)\phi_{j}(t)|=O_{p}(n^{-1/2}\kappa_{j}).
\]
And similarly we can show that $\hat{\Delta}_{jk}=O_{p}(n^{-1/2}\kappa_{j}^{1/2}\kappa_{k}^{1/2})$
for any $j\neq k$, which will be used later in calculating the order
of the fourth term. The second term becomes 
\begin{align*}
 & \int\int Q^{+}(s,t)(\tilde{Q}-Q)(s,t)dsdt\\
 & \quad=\sum_{k=1}^{\infty}\frac{1}{\lambda+\kappa_{k}}\int\int\Delta(s,t)\phi_{j}(s)\phi_{j}(t)dsdt\\
 & \quad\leq O_{p}(n^{-1/2}\sum_{k=1}^{\infty}\frac{\kappa_{k}}{\lambda+\kappa_{k}}).
\end{align*}

For the third term, we refer to (6.7) of \cite{hall_05} that $||(I+Q^{+}\Delta)^{-1}||=O_{p}(1)$.
Here $||\cdot||$ as a norm of a functional from $L_{2}[0,1]$ to
itself, is defined as 
\[
||\chi||=\sup_{\phi\in L_{2}[0,1],||\phi||=1}||\chi(\phi)||.
\]
Noting that $\tilde{Q}^{+}-Q^{+}=-(I+Q^{+}\Delta)^{-1}Q^{+}\Delta Q^{+}$,
then
\begin{align*}
 & \int\int(\tilde{Q}^{+}-Q^{+})(s,t)Q(s,t)dsdt\\
 & \quad=-\sum_{k=1}^{\infty}\kappa_{k}\int\int(I+Q^{+}\Delta)^{-1}Q^{+}\Delta Q^{+}(s,t)\phi_{k}(s)\phi_{k}(t)dsdt\\
 & \quad=-\sum_{k=1}^{\infty}\frac{\kappa_{k}}{\lambda+\kappa_{k}}\int\int(I+Q^{+}\Delta)^{-1}Q^{+}\Delta(s,t)\phi_{k}(s)\phi_{k}(t)dsdt\\
 & \quad\leq\sum_{k=1}^{\infty}\frac{\kappa_{k}}{\lambda+\kappa_{k}}||(I+Q^{+}\Delta)^{-1}Q^{+}\Delta(s,t)\phi_{k}(s)||\,||\phi_{k}(t)||\\
 & \quad=O_{p}(\sum_{k=1}^{\infty}\frac{\kappa_{k}}{\lambda+\kappa_{k}}).
\end{align*}
The last equation follows from the fact that 
\begin{align*}
 & ||(I+Q^{+}\Delta)^{-1}Q^{+}\Delta(s,t)\phi_{k}(s)||\\
 & \quad=||\phi_{k}(s)-(I+Q^{+}\Delta)^{-1}\phi_{k}(s)||\\
 & \quad\leq||\phi_{k}(s)||+||(I+Q^{+}\Delta)^{-1}||\,||\phi_{k}(s)||\\
 & \quad=1+||(I+Q^{+}\Delta)^{-1}||.
\end{align*}

For the last term, by Cauchy-Schwarz inequality
\begin{align*}
\int\int(\tilde{Q}-Q) & (s,t)(\tilde{Q}^{+}-Q^{+})(s,t)dsdt\\
 & \leq\big\{\int\int(\tilde{Q}-Q)^{2}(s,t)dsdt\cdot\int\int(\tilde{Q}^{+}-Q^{+})^{2}(s,t)dsdt\big\}^{1/2}\\
 & \leq n^{-1/2}\big\{\int\int(\tilde{Q}^{+}-Q^{+})^{2}(s,t)dsdt\big\}^{1/2}\\
 & =n^{-1/2}\big\{\sum_{k=1}^{\infty}||-(I+Q^{+}\Delta)^{-1}Q^{+}\Delta Q^{+}\phi_{k}||^{2}\big\}^{1/2}\\
 & \leq n^{-1/2}||(I+Q^{+}\Delta)^{-1}||^{-1/2}\big\{\sum_{k=1}^{\infty}||Q^{+}\Delta Q^{+}\phi_{k}||^{2}\big\}^{1/2}.
\end{align*}
Recall that $\hat{\Delta}_{jk}=O_{p}(n^{-1/2}\kappa_{j}^{1/2}\kappa_{k}^{1/2})$
for any $j\neq k$. Then, 
\begin{align*}
||Q^{+}\Delta Q^{+}\phi_{k}||^{2} & =\int\big\{\int\int Q^{+}\Delta(s,u)\sum_{j=1}^{\infty}\frac{1}{\lambda+\kappa_{j}}\phi_{j}(u)\phi_{j}(t)\phi_{k}(t)dtdu\big\}^{2}ds\\
 & =\frac{1}{(\lambda+\kappa_{k})^{2}}\int\big\{\int Q^{+}\Delta(s,u)\phi_{k}(u)du\big\}^{2}ds\\
 & =\frac{1}{(\lambda+\kappa_{k})^{2}}\int\big\{\int\int\sum_{j=1}^{\infty}\frac{1}{\lambda+\kappa_{j}}\phi_{j}(s)\phi_{j}(v)\Delta(v,u)\phi_{k}(u)dudv\big\}^{2}ds\\
 & =\frac{1}{(\lambda+\kappa_{k})^{2}}\sum_{j=1}^{\infty}\frac{\hat{\Delta}_{jk}^{2}}{(\lambda+\kappa_{j})^{2}}\\
 & =O_{p}(\frac{\kappa_{k}}{(\lambda+\kappa_{k})^{2}}n^{-1}\sum_{j=1}^{\infty}\frac{\kappa_{j}}{(\lambda+\kappa_{j})^{2}})\\
 & =O_{p}\big(n^{-1}\lambda^{-2}(\sum_{j=1}^{\infty}\frac{\kappa_{j}}{\lambda+\kappa_{j}})\frac{\kappa_{k}}{\lambda+\kappa_{k}}\big).
\end{align*}
Therefore
\[
\int\int(\tilde{Q}-Q)(s,t)(\tilde{Q}^{+}-Q^{+})(s,t)dsdt=O_{p}(n^{-1}\lambda^{-1}\sum_{j=1}^{\infty}\frac{\kappa_{j}}{\lambda+\kappa_{j}}).
\]
 
All together we show that $\sum_{k=1}^{\infty}\frac{\tilde{\kappa}_{k}}{\lambda+\tilde{\kappa}_{k}}=O_{p}(\sum_{k=1}^{\infty}\frac{\kappa_{k}}{\lambda+\kappa_{k}})$
provided that $\lambda^{-1}=O(n)$.
\end{proof}

\bibliographystyle{myrefstyle}
\bibliography{TestReference}

\end{document}